\documentclass[envcountsame,orivec,runningheads,11pt]{llncs}
\usepackage{a4,graphicx,xspace,amsmath,amssymb,amsfonts,stmaryrd,textgreek,xcolor}
\usepackage{hyperref,breakurl}             % Not needed if you use pdflatex only.
\usepackage[normalem]{ulem}
\usepackage{bm} %bold math
\usepackage[capitalise]{cleveref}
\usepackage{todonotes}
\usepackage{orcidlink}
\newcommand{\myprime}{\textquotesingle}
\newcommand{\myprimes}{\myprime\myprime}
\newcommand{\seq}[1]{\mathbf{#1}}
\DeclareFontFamily{U}{mathx}{\hyphenchar\font45}
\DeclareFontShape{U}{mathx}{m}{n}{<-> mathx10}{}
\DeclareSymbolFont{mathx}{U}{mathx}{m}{n}
\DeclareMathAccent{\widebar}{0}{mathx}{"73}
\newcommand{\IR}{\mathbb{R}}
\newcommand{\IC}{\mathbb{C}}

\newcommand{\IE}{\mathbb{E}}

\newcommand{\IZ}{\mathbb{Z}}
\newcommand{\IN}{\mathbb{N}}
\newcommand{\Log}{\operatorname{Log}}
\newcommand{\poly}{\operatorname{poly}}

\newcommand{\Card}{\operatorname{Card}}

\newcommand{\ELL}[1]{\ensuremath{\mathrm{L}^{#1}}}
\newcommand{\Ell}[1]{\ensuremath{\ell^{#1}}}

\newcommand{\ELLtwo}{\ELL{2}}

\newcommand{\Elltwo}{\Ell{2}}
\newcommand{\SOB}[1]{\ensuremath{\mathrm{H}^{#1}}}
\newcommand{\IT}{\mathbb T}
\newcommand{\N}{\mathbb N}

\newcommand{\Norm}[1]{{\left\vert\kern-0.25ex\left\vert\kern-0.25ex\left\vert #1 \right\vert\kern-0.25ex\right\vert\kern-0.25ex\right\vert}} %operator norm with 3 vertical bars
\newcommand{\norm}[1]{\left\lVert #1\right\rVert}
\newcommand{\abs}[1]{\left\lvert #1\right\rvert}

\newcommand{\bin}{\mathrm{bin}}
\newcommand{\id}{\operatorname{id}}
\renewcommand{\Re}{\operatorname{Re}}
\renewcommand{\Im}{\operatorname{Im}}
\newcommand{\sdzero}{\textup{\texttt{0}}\xspace}
\newcommand{\sdone}{\textup{\texttt{1}}\xspace}

\newcommand{\TWO}{\{\sdzero,\sdone\}}
\newcommand{\Algorithm}{\calA}
\newcommand{\calA}{\mathcal{A}}
\newcommand{\calB}{\mathcal{B}}
\newcommand{\calC}{\mathcal{C}}
\newcommand{\calF}{\mathcal{F}}

\newcommand{\calO}{\mathcal{O}}
\newcommand{\calL}{\mathcal{L}}
\newcommand{\calK}{\mathcal{K}}
\newcommand{\calM}{\mathcal{M}}
\newcommand{\calP}{\mathcal{P}}
\newcommand{\calS}{\mathcal{S}}
\newcommand{\calT}{\mathcal{T}}

\newcommand{\classP}{\text{\rm\textsf{P}}\xspace}

\newcommand{\classNP}{\text{\rm\textsf{NP}}\xspace}
\newcommand{\classFP}{\text{\rm\textsf{FP}}\xspace}

\newcommand{\classSharpP}{\text{\rm\textsf{\#P}}\xspace}
\newcommand{\classGapP}{\text{\rm\textsf{GapP}}\xspace}

\newcommand{\classPSPACE}{\text{\rm\textsf{PSPACE}}\xspace}
\newcommand{\classEXP}{\text{\rm\textsf{EXP}}\xspace}
\newcommand{\classNEXP}{\text{\rm\textsf{NEXP}}\xspace}
\newcommand{\runtime}{\text{\rm time}\xspace}
\newcommand{\precision}{\text{\rm prec}\xspace}

\newcommand{\ie}{i.\,e.\xspace}

\newtheorem{observation}[theorem]{Observation} %attention: cleveref \cref somehow does not work here; use amsthm?
\newtheorem{fact}[theorem]{Fact}  %attention: cleveref \cref somehow does not work here; use amsthm?
\crefname{fact}{fact}{facts}\Crefname{fact}{Fact}{Facts}

\makeatletter
\DeclareRobustCommand*{\bfseries}{%
  \not@math@alphabet\bfseries\mathbf
  \fontseries\bfdefault\selectfont
  \boldmath
}
\makeatother

\newcommand{\Martin}[1]{} %comments in text
\newcommand{\Aras}[1]{}
\newcommand{\Sveta}[1]{}
\newcommand{\COMMENTED}[1]{}

\title{Complexity Theory for Integrable Functions}
\author{Aras Bacho$^1$ \and Svetlana Selivanova$^2$ \and Martin Ziegler$^3$}  %ALPHABETICAL order!!
\authorrunning{Bacho, Selivanova, Ziegler}
\institute{$^1$ California Institute of Technology, 
$^2$ St. Petersburg State University and \\ A.P. Ershov Institute of Informatics Systems SB RAS, $^3$ KAIST}
\date{}

\begin{document}
\setcounter{page}{0}
\maketitle
\begin{abstract}
Initiated by Alan Turing, Computable Analysis has traditionally focused on real functions which are continuous:
exploring for example which (e.g. solution) operators preserve polynomial-time/computability, and which ones do not.
However, partial differential equations naturally operate on spaces of integrable rather than continuous functions.
Here, several canonical definitions of qualitative computability have been shown equivalent.
Towards qualitative complexity investigations,
we formalize two (plus one) natural conceptions of polynomial-time computation for $\ELLtwo$ (=square-integrable) functions.
And we show them incomparable, unless $\classSharpP_1\subseteq\classFP_1$.
Also the Heat and the Wave Equation behave differently.
\end{abstract}

\setcounter{secnumdepth}{3}
\setcounter{tocdepth}{3}
\renewcommand{\contentsname}{}
%\textbf{\small Table of Contents:} \nopagebreak
%\section*{Table of Contents}
\begin{center}\small
\begin{minipage}[c]{0.98\textwidth}\vspace*{-8ex}%
\tableofcontents
\end{minipage}
\end{center}

\pagebreak
%%%%%%%%%%%%%%%%%%%%%%%%%%%%%%%%%%%%
\section{Introduction}

\Martin{[Now use \texttt{\textbackslash hat} for Fourier, as suggested by @Aras; 
\texttt{\textbackslash bar} for averages, \texttt{\textbackslash seq} for sequences]}
Is a certain numerical algorithm optimal or can it be improved, asymptotically?
How does the computational cost intrinsic to solving non-linear ordinary differential equations (ODEs) 
compare to that of solving partial but linear PDEs?
Formalizing and answering such questions in a sound sense
is the subject of complexity theory.
%, which traditionally focuses on discrete/ized problems.
For example maximizing a fixed continuous---and even a smooth---polynomial-time computable real function
has been proven to correspond to complexity class $\classNP_1$
\cite[Theorem~3.16]{Ko91};
whereas in/definite integration corresponds to the 
(presumed strictly larger)
complexity classes $\classSharpP$ and $\classSharpP_1$, respectively:
again in both the continuous and the smooth case
\cite[Theorems~5.32+5.33]{Ko91}. 
And solving Lipschitz-continuous ODEs (which generalizes indefinite integration) 
corresponds to the even larger class $\classPSPACE$ \cite{Kaw10}.
Similar investigations have related various linear PDEs to common discrete complexity classes 
\cite{KSZ17,DBLP:journals/tsp/BocheP21,KPSZ23}, but could consider only classical
(\ie continuous/ly differentiable) solutions: 
due to lack of an established definition for polynomial-time (say, square-) integrable functions.
Recall that many PDEs naturally operate on spaces of integrable
rather than continuous functions, also computably \cite[Theorems~3.8+3.11+4.5]{WZ02a}.
Regarding computability of integrable functions,
several natural definitions have turned out as equivalent \cite{ZZ99}.
The present work initially started out trying to quantitatively strengthen 
said qualitative equivalence to (polynomial-time) complexity \cite{CCA2020}, 
but came to realize and investigate the opposite!
It is based on conference versions \cite{Aras25,Sveta26},
combined and tailored and extended and self-contained.

%%%%%%%%%
\subsection{Well-Defining Computational Worst-Case Cost}
\label{ss:Tradition}

The computational cost of an algorithm $\Algorithm$ receiving binary strings as input 
is commonly measured 
in dependence on the length $n$ of said input, more precisely: 
When $t_\Algorithm(x)<\infty$ denotes the number of steps that $\Algorithm$ makes
on input $x\in\TWO^*$, one usually considers its worst-case runtime 
$t_\Algorithm(n) = \max\big\{t_\Algorithm(x):|x|\leq n\big\}$, %\end{equation}
which is obviously well-defined (namely pointwise finite) and a non-decreasing function of type $\IN\nearrow\IN$, 
its asymptotic growth being the core subject of classical complexity theory.
This well-definition of worst-case cost 
\begin{equation}
\label{e:Worstcase0}
t_\Algorithm(n) \;=\;  \max\big\{t_\Algorithm(x):x\in X_n\big\}
\end{equation}
holds trivially as well for computations over any countable domain $X$,
when equipped with a monotone cover $X_n$ by \emph{finite} subsets of $X$:
for instance $X_n=\big\{ \text{graph } G=(V,E) \text{ with } |V|\leq n \text{ vertices}\big\}$.

Perhaps less known, but folklore in Computable Analysis \cite[Theorems~7.1.5 +7.2.7]{Wei00}:
\Cref{e:Worstcase0} remains well-defined when generalizing `finite' to `compact';
specifically when the domain $X=\bigcup_n X_n$ is a suitably encoded \cite{Sch04}
$\sigma$-compact topological space with compact cover $X_n\subseteq X_{n+1}$.
For instance $t_\Algorithm(n)=\max\big\{t_\Algorithm(x,n):0\leq x\leq 1\big\}$ is still well-defined
when $\Algorithm$ computes
a continuous total function $f:[0;1]\to\IC$ in the usual sense of \cite[Definition~2.11]{Ko91} or \cite[Example~7.2.14]{Wei00},
namely pointwise approximating $x\mapsto f(x)$ up to error $2^{-n}$.
This well-definition is crucial to the rich complexity theory of real (and complex) continuous functions \cite{Ko91,BC06},
including investigations on the complexity blow-up 
\cite{Fri84,Kaw10,KSZ17,DBLP:journals/tsp/BocheP21,KPSZ23,DBLP:journals/tit/BochePP24,DBLP:journals/tcom/BocheGSP24}
incurred by popular operators and functionals on suitable spaces of (uniformly) continuous functions.
Indeed, Computable Analysis has long focused on continuous functions only \cite{Grz57}:
because any terminating computation $\Algorithm(x)$ by itself depends continuously on its input $x$,
namely $\Algorithm$ can read only a finite initial part of the infinite information comprising $x\in X$.
%This is essentially what makes the maximum in \cref{e:Worstcase0} finite.

%%%%%%%%%%%%%%%%%%%%%%%%
\subsection{Overview}
\label{ss:Overview}

Integrable functions arise naturally in the study of Partial Differential Equations (PDEs) \cite{Brez11FASS}.
Indeed, the Banach spaces $\ELL{p}$ ($1\leq p\leq\infty$)
generalize and contain the space $\calC$ of continuous complex functions.
We focus on the arguably best-behaved case 
$p=2$ of square-integrable functions, 
since it constitutes a separable Hilbert space.
A $p$-integrable function $f:X\to\IC$ is technically an equivalence 
class of mappings which differ on some subset of $X$ of measure zero.
One thus cannot reasonably define computing $f$ in terms of algorithmic evaluation $x\mapsto f(x)$. 
Instead, the literature knows various other (well-)definitions 
of qualitative computability for integrable functions, and has shown them equivalent \cite{PR89,ZZ99,DBLP:journals/mlq/BrattkaD07}.
Previous works have also proposed quantitative notions of computational cost \cite[\S3.2]{Ste17},\cite[\S4]{Kun04a}.
We capture and formalize such notions, based on suitable guidance from the continuous case;
and we prove them equivalent with respect to $\classSharpP$ complexity,
but incomparable under the finer polynomial-time perspective---unless $\classSharpP_1\subseteq\classFP_1$.

\medskip
Write $\poly(n)=\calO\big(n^{\calO(1)}\big)$ for polynomially bounded asymptotic growth,
and $\TWO^*$ for the set of finite binary strings.
Recall that $\classSharpP\supseteq\classFP$ denotes the class of polynomial witness counting problems 
\begin{equation} \label{e:SharpP}
\varphi(\vec v) \;=\; \Card\big\{\vec w\in\TWO^*: |\vec w|\leq\poly(|\vec v|), \vec w\in P\big\} \;\in\IN, \quad P\in\classP
\end{equation}
and $\classGapP\supseteq\classSharpP$ its closure under negation/subtraction \cite{GapP}.
Intuitively $\classSharpP$ lies between classes $\classNP$ and $\classPSPACE$,
although technically the latter two consist of decision and the former of integer function problems.
We frequently suppose $\classSharpP_1\not\subseteq\classFP_1$, 
where subscript 1 indicates restriction to unary inputs \cite[\S5]{Valiant79}.
A weaker variant of the \emph{Millennium Prize Problem} asks whether $\classP_1=\classNP_1$,
equivalently: whether $\classEXP=\classNEXP$ \cite[Exercises~2.30+2.32]{Arora09}.

\bigskip
\Cref{ss:RecapC} recalls common both computability
and complexity (polynomial-time and \classSharpP) notions of computation for continuous functions. 
\Cref{s:Main} formalizes two previous, plus one new, natural counterparts (a,b,c) for \ELLtwo{} functions:
(c) pertains to \emph{average-case} complexity, as  elaborated in \cref{ss:Average}.
\Cref{ss:Comparing} compares these notions,
exhibiting both implications and counterexamples %lack thereof
(unless $\classSharpP_1\subseteq\classFP_1$): 
also in relation to the continuous case from \cref{ss:RecapC}.
\Cref{ss:Closure} collects various closure properties of (a,b,c), 
such as under shifts and differentiation and indefinite integration.

As applications (\cref{s:Applications}), 
\Cref{ss:Heat} treats the Heat Equation 
from our new perspective: 
With respect to Fourier computation (a),
it turns out to map polynomial-time initial conditions
to polynomial-time solutions.
This contrasts the classical setting,
see \cref{r:Blowup}.
The Wave Equation in \cref{ss:Wave} 
exhibits a similar effect: preserving 
(a) Fourier polynomial-time computability,
while $\classSharpP_1$-hard
w.r.t. classical (=continuous/ly differentiable) solutions.

%%%%%%%%%%%%%%%%%%%%%%%%
\subsection{Recap: Polynomial-Time Complexity of Continuous Functions}
\label{ss:RecapC}

\Cref{ss:Tradition} has already mentioned one common 
conception of computing a continuous function $f:[0;1]\to\IC$,
which however does not generalize to \ELLtwo,
namely based on algorithmic evaluation $t\mapsto f(t)$:
Given any (suitable sequence of approximations to) $t\in[0;1]$ and given $n\in\IN$, 
output some Gaussian dyadic rational $(u+iv)/2^n$ 
approximating $f(t)$ up to absolute error $2^{-n}$.
\Cref{d:One}b) below adopts an alternative but polynomial-time equivalent 
characterization from \cite[Theorem~2.22]{Ko91}.

Another computably equivalent definition for $f\in\calC[0;1]$,
based on the Weierstra\ss{} approximation theorem 
and captured in \Cref{d:One}a) below,
requires algorithmic output only: 
of (the coefficient vectors of) a sequence of Gaussian dyadic rational polynomials $p_n$ 
approximating $f$ in the $\infty$-norm up to error $2^{-n}$ \cite[\S I.0.3]{PR89}.
Indeed, the set $Q$ of such polynomials is countable, 
and dense in $\calC[0;1]$ by Stone-Weierstra\ss{}.
This suggests \cite[Definition~8.1.2]{Wei00} 
defining computation of a point $x$ from some fixed separable metric space $(X,d)$
in terms of algorithmically approximating $x$ up to error $2^{-n}$ by a suitable sequence
of elements from some fixed countable dense subset $Q\subseteq X$;
see also \cite{LimZiegler25}.
However for the space $\calC[0;1]$ this alternative definition, although qualitatively equivalent, 
can differ exponentially in quantitative computational cost, see \cite[\S8]{Ko91}. 

\begin{definition}
\label{d:One}
The (binary) \emph{modulus} of (uniform) continuity of $f\in\calC[0;1]$
is the non-decreasing mapping $\mu_f:\IN\nearrow\IN$ defined by
\begin{equation}
\label{e:Modulus}
n\;\mapsto\; \min\big\{ m\in\IN: 
\forall t,t'\in[0;1]: \big(|t-t'|\leq2^{-m} \Rightarrow |f(t)-f(t')|\leq2^{-n} \big)\big\}
%\forall 0\leq t\leq t'\leq 1, |t-t'|\leq2^{-m} \Rightarrow |f(t)-f(t')|\leq2^{-n} \big\}
%\enspace .
\end{equation}
\Martin{[Changed from `\emph{a}' to \emph{`the' (least)} modulus,
to be/come consistent with \cref{d:Moduli}.
Also, use \emph{evaluable} instead of \emph{computable} for continuous functions.]}
Let $\calP_D$ denote the vector space of complex univariate polynomials of degree $\leq D$.
Call $f\in\calC[0;1]$ \emph{evaluable in polynomial time} ~iff~
there exists a Gaussian integer double sequence $\seq v=(v_{n,L,\ell})\in\IZ+i\IZ$, 
indexed by $n\in\IN$ and $0\leq \ell<L\in\IN$ such that (i), (ii), and (iii) hold:
\begin{description}
\item[i)] $\forall n,\ell,L: \; 
\big|f(\ell/L)-v_{n,L,\ell}/2^n\big|\:\leq\:2^{-n}$
\item[ii)] Both $(n,L,\ell)\mapsto \Re v_{n,L,\ell}\in\IZ$
~and~ $(n,L,\ell)\mapsto\Im v_{n,L,\ell}\in\IZ$
\newline
are computable in time $\poly(n+\log L)$.
\item[iii)] $f$ has \emph{polynomially bounded} modulus of continuity $\mu_f(n)\leq\poly(n)$.
\smallskip
\end{description}
Call $f$ \emph{evaluable in \classSharpP} if (i), (iii), and (ii\,\myprime) hold:
\begin{description}
\item[ii\,\myprime)] 
Both $\TWO^*\:\ni\:\sdone^n\,\sdzero\,\bin(L)\,\bin(\ell)\mapsto \Re v_{n,L,\ell}\in\IZ$~ and
\newline
$\TWO^*\:\ni\:\sdone^n\,\sdzero\,\bin(L)\,\bin(\ell)\mapsto \Im v_{n,L,\ell}\in\IZ$~ lie in \classGapP.
\end{description}
\smallskip
Call $f$ \emph{computably evaluable} if (i), (ii\,\myprimes), and (iii\,\myprimes) hold:
\begin{description}
\item[ii\,\myprimes)] 
Both $\Re v_{n,L,\ell},\Im v_{n,L,\ell}\in\IZ$ from (ii)/(ii\,\myprime) are recursive, 
\newline \ie computable in the discrete sense.
\item[iii\,\myprimes)]
$f$ has a recursively (=computably) bounded modulus of continuity $\mu_f$.
\qed\end{description}
\end{definition} 
It is customary in complexity theory to consider
integer parameters as exponents to base $2^{\pm1}$,
such as $m$ and $\mu$ in the modulus;
see \cref{r:Unary} below for further discussion
of `binary' versus `unary' parameterizations.
Binary modulus $\mu_f$ roughly means that 
$\approx\mu_f(n)$ bits of the input approximation are sufficient
to guarantee $\approx n$ correct bits of the output approximation:

\begin{fact}
\label{f:Modulus}
\begin{enumerate}
\item[a)]
If $f$ can be approximated up to error $2^{-n}$
in time $t(n)$, then its modulus of continuity satisfies 
$\mu_f(n)\leq t(n+2)$ {\rm\cite[Theorem~2.19]{Ko91}}.
\item[b)]
Conversely, if $f$ has modulus of continuity $\mu_f$, 
then there exists some (oracle and an) oracle Turing machine 
computing $f$ in time $\calO\Big(\mu_f\big(\calO(n)\big)\Big)$
{\rm\cite[Corollary~3.6b]{LimZiegler25}}.
\item[c)]
A function $f$ is H\"{o}lder-continuous ~iff~
its modulus of continuity is linearly bounded: $\mu_f(n)\leq\calO(n)$.
It is Lipschitz-continuous ~iff~
its modulus of continuity satisfies $\mu_f(n)\leq n+\calO(1)$.
\item[d)]
The function $h:[0;1]\to [0;1], \;h(t)=1/\ln(e/t)$ %,\; h(0)=0,$ 
has exponential, but not polynomial, modulus of continuity.
Iterate $h\circ h:[0;1]\to [0;1]$ has \emph{doubly},
but not single, exponential modulus of continuity; and so on.
\end{enumerate}
\end{fact}
\noindent
Note that $\calO(n+\log L)$ is the input length in \cref{d:One}(ii\myprime), 
which thus generalizes (ii) meaning that both mappings 
$\sdone^n\,\sdzero\,\bin(L)\,\bin(\ell)\mapsto \Re v_{n,L,\ell},\Im v_{n,L,\ell}\in\IZ$ 
belong to \classFP. 
Reasonably enough,
relaxing discrete polynomial complexity requirements (ii+iii) to being merely recursive (ii\myprimes+iii\myprimes) 
makes \cref{d:One} boil down to one of the many equivalent conceptions
of a computable continuous function \cite[\S I.0]{PR89}.
Our conception (ii\myprime) of computability in $\classSharpP$ here may be stronger than
\cite[Definitions~13+14]{KPSZ23}: 
the latter would allow for (repeated and adaptive and post-processable)
oracle access to $\classSharpP$ while the present work takes care 
not to climb up the Counting Hierarchy \cite{Wagner}, 
such as in \cref{t:SharpP,t:Continuous1,t:Continuous3} below.

%%%%%%%%%%%%%%%%%%%%%%%%%%%%%%%%%%%%
\section{Polynomial-Time Complexity for \ELLtwo-Functions}
\label{s:Main}

Guided by the case of continuous functions in \cref{ss:RecapC},
this section formalizes (\cref{ss:Candidates}) and compares 
(\cref{ss:Comparing}) notions
of polynomial-time complexity for square-integrable functions.

\medskip
%\subsection{Recap of $\ELL{p}$ Spaces} \label{ss:Recap}

Let $\IT=[0;1)\bmod 1$ denote the unit interval with wraparound.
For $1\leq p\leq\infty$, 
$\ELL{p}(\IT)$ is the space of equivalence classes of Lebesgue measurable 
complex-valued functions $f:\IT\to\IC$
which are $p$-integrable: $\|f\|_p<\infty$,
with the $p$-norm $\|f\|_p=\sqrt[p]{\int |f(t)|^p\,dt}$.

Unlike for $f\in\calC(\IT)$, 
point evaluation $\IT\ni t\mapsto f(t)\in\IC$ 
is in general not well-defined for $f\in\ELL{p}(\IT)$.
Like $\calC(\IT)$, $\ELL{p}(\IT)$ is separable:
for $1\leq p<\infty$, but not for $p=\infty$.
For $1<p<q<\infty$, record the hierarchy
\begin{equation}
\label{e:Embeddings}
\calC(\IT) \;\subsetneq\; \ELL{\infty}(\IT) \;\subsetneq\; 
\ELL{q}(\IT) \;\subsetneq\; 
\ELL{p}(\IT) \;\subsetneq\; 
\ELL{1}(\IT) , \quad 
\|f\|_p\leq \|f\|_q \enspace .
\end{equation}
When replacing $\IT$ with some other domain $X$ of finite measure,
the norm inequality includes a constant factor: $\|f\|_p\leq \calO(\|f\|_q)$.
We focus on the arguably best-behaved case $p=2$:
$\ELLtwo(\IT)$ constitutes a separable Hilbert space
with respect to the inner product
$\langle f,g\rangle  = \int\nolimits_0^1 f(t) \cdot g^*(t)\,dt$.

\medskip
$\ELLtwo(\IT)$ is isometric, according to Parseval,
to the space $\Elltwo$ of square-summable complex sequences $(\hat f_k)_{k\in\IZ}$
of Fourier coefficients 
\begin{equation}
\label{e:Fourier1}
\hat f_k \;:=\; \int\nolimits_0^1 \exp(-2\pi ik t) \cdot f(t)\,dt
\enspace .
\end{equation}
The latter are well-defined whenever $f\in\ELL{p}(\IT)$, $p\geq1$;
and unique \cite[Corollary~1.11(2)]{Duoandikoetxea}.
Let 
\begin{equation}
\label{e:Fourier}
\calF_K \;:=\;
\big\{ \sum\nolimits_{|k|\leq K} c_k\cdot\exp(2\pi i k t): c_k\in\IC \big\}
\;\subseteq\;\calC^\infty(\IT) 
\end{equation}
denote the space of trigonometric polynomials of degree $\leq K$.
Write $\calF_K(f,t):=\sum_{|k|\leq K} \hat f_k\cdot\exp(2\pi ik t)\in\calF_K$ 
for the first $2K+1$ terms of $f$'s Fourier series.
For fixed $1<p<\infty$ and $f\in\ELL{p}$, $\calF_K(f)$
converges in $p$-norm to $f$ by the Marcel Riesz Theorem,
but not not necessarily for $p=1$ nor $p=\infty$ \cite[\S1.4]{Duoandikoetxea};
see \cref{eq:uniform-fourier-bound} below.
We also use $\calS_S$ to denote the space of 
complex-valued functions on $[0;1]$ that are piecewise constant
on each of the equidistant half-open intervals $\big[\tfrac{s}{S};\tfrac{s+1}{S}\big)$,
$s=0,\ldots,S-1$. To $f\in\ELL{1}(\IT)$ assign
\begin{equation}
\label{e:StepAverage} 
\calS_S(f)\in\calS_S, \qquad \calS_S(f,t)\;:=\;\bar f_{S,s}, \quad \tfrac{s}{S}\leq t<\tfrac{s+1}{S} 
\end{equation}
whose piecewise heights are given by averages $\bar f_{S,s}:=S\cdot\int_{s/S}^{(s+1)/S} f(t)\,dt$.
As $S\to\infty$, $\calS_S(f)$ converges to $f\in\ELL{p}(\IT)$ in $p$-norm.

\begin{remark}
\label{r:BestApprox}
Note that $\hat f_0=\bar f_{1,0}$.

Within the space $\calF_K$, 
$\calF_K(f)$ is the best approximation to $f$
with respect to $\norm{\cdot}_2$;
but not necessarily 
with respect to $\norm{\cdot}_p$
for $p\neq2$, yet good enough
according to the proof of \cref{t:Aras}a).
The best approximation in $\calF_K$ is unique for $1<p<\infty$. 

Similarly, within the space $\calS_S$,
$\calS_S(f)$ is the best approximation to $f$
with respect to $\norm{\cdot}_2$;
but not necessarily 
with respect to $\norm{\cdot}_p$
for $p\neq2$, yet good enough
according to the proof of \cref{t:Aras}c+e).
The best approximation in $\calF_S$ is unique for $1<p<\infty$.
\end{remark}

%%%%%%%%%%%%%%%%%%
\subsection{Candidate Definitions}
\label{ss:Candidates} 

Calculus knows (at least) two common choices of countable dense subsets of $\ELLtwo[0;1]$:
a) Fourier series with Gaussian dyadic rational coefficients ~and~ 
b) piecewise constant functions comprised of equidistant blocks of Gaussian dyadic rational heights.
This suggests the following two (plus one) ways of generalizing \cref{d:One}.

\begin{definition}
\label{d:Two}
\begin{description}
\item[a)]
Call $f\in\ELLtwo(\IT)$ \emph{Fourier-}computable in polynomial time ~iff
there exists a Gaussian integer double sequence $\seq c=(c_{k,n})\in\IZ+i\IZ$, 
indexed by $n\in\IN$ and $k\in\IZ$, such that (a\,i), (a\,ii), and (a\,iii) hold:
%\begin{description}
\item[a\,i)] $\forall n,k: \; \big|\hat f_k-c_{k,n}/2^n\big|\:\leq\:2^{-n}$   
with $f$'s Fourier coefficients $\hat f_k$
\item[a\,ii)] Both $(n,k)\mapsto \Re c_{k,n}\in\IZ$
~and~ $(n,k)\mapsto\Im c_{k,n}\in\IZ$
\newline are computable in time $\poly(n+\log|k|)$.
\item[a\,iii)] It holds $\big\| f - \calF_{K}(f)\big\|_2 \leq2^{-m}$ for $K=2^{\poly(m)}$.
\medskip
\item[a\myprime)]
Call $f\in\ELLtwo(\IT)$ Fourier-computable in \classSharpP ~iff~ (a\,i), (a\,iii), and (a\myprime\,ii) hold:
%\begin{description}
\item[a\myprime\,ii)] 
Both $\TWO^*\:\ni\:\sdone^n\,\sdzero\,\bin(k)\mapsto \Re c_{k,n},\Im c_{k,n}\in\IZ$ belong to \classGapP.
\medskip
\item[b)]
Call $f\in\ELLtwo(\IT)$ \emph{Step-}computable in polynomial time ~iff~
there exists a Gaussian integer multi sequence $\seq h=(h_{S,s,n})$, indexed by $n\in\IN$ and $0\leq s<S\in\IN$, such that (b\,i), (b\,ii), and (b\,iii) hold:
%\begin{description} XXX S_n only (all/some) powers of 2
\item[b\,i)] $\forall n,S,s: \; 
\big|\bar f_{S,s}-h_{S,s,n}/2^n\big|\:\leq\:2^{-n}$ with averages $\bar f_{S,s}$ according to \cref{e:StepAverage}.
\item[b\,ii)] $(S,s,n)\mapsto \Re h_{S,s,n},\Im h_{S,s,n}\in\IZ$ are computable in time $\poly(n+\log S)$.
\item[b\,iii)] It holds $\big\| f- \calS_S(f)\big\|_2 \leq2^{-m}$ for $S=2^{\poly(m)}$.
\medskip
\item[b\myprime)]
Call $f\in\ELLtwo(\IT)$ Step-computable in \classSharpP ~iff~ 
(b\,i), (b\,iii), and (b\myprime\,ii) hold:
\item[b\myprime\,ii)] 
$\TWO^*\:\ni\:\sdone^n\,\sdzero\,\bin(S)\,\bin(s)\mapsto \Re h_{S,s,n},\Im h_{S,s,n}\in\IZ$ belong to \classGapP.
\medskip
\item[c)]
Call $f\in\ELLtwo(\IT)$ Step-computable in polynomial time on \emph{average} ~iff~
for some $S(n)=2^{\poly(n)}$, 
there exists a Gaussian integer multi sequence $\seq h=(h_{s,n})=(h_{S(n),s,n})$ 
such that (c\,i), (b\,ii), (b\,iii) hold:
\item[c\,i)] $\sqrt{\sum_{s=0}^{S(n)-1} \big|\bar f_{S(n),s}-h_{s,n}/2^n\big|^2/S(n)} \;\leq 2^{-n}$.
\item[c\myprime)]
Call $f\in\ELLtwo(\IT)$ Step-computable in \classSharpP on \emph{average} ~iff~
(c\,i), (b\myprime\,ii), and (b\,iii) hold.
\end{description}
\end{definition}
Intuitively, (a\,iii) means that said Fourier series converges `generalized logarithmically';
see \cref{ss:Moduli} below.
Analogously, (b\,iii) means that step functions converge generalized logarithmically.
Note also that (iii) and (a\,iii) and (b\,iii) are mathematical conditions,
while (ii)/(ii\myprime) and variants (a\,ii)/(a\myprime\,ii)/(b\,ii)/(b\myprime\,ii) pertain to computation.
\Cref{ss:Aras} below proves (a\,iii) and (b\,iii) equivalent, and that (iii) implies (a/b\,iii).
And \cref{ss:Average} below elaborates on \cref{d:Two}c).

Similarly to \cref{d:One}(ii\myprimes+iii\myprimes),
weakening complexity conditions (a\,ii+iii)/ (b\,ii+iii) to mere recursivity
makes \cref{d:Two} boil down to the many equivalent conceptions
of \emph{computable} \ELLtwo{} function \cite[\S II.1]{PR89,ZZ99,Kun04a};
see \cref{c:SharpP}.

%%%%%%%%%%%%%%%%%%%%%%%
\subsection{Moduli and Rates of Approximation}
\label{ss:Moduli}

In \cref{d:One}, having polynomially bounded (binary) modulus of continuity
is an important prerequisite for a continuous function to be polynomial-time computable.
For an $\ELLtwo$ function,
\cref{d:Two}a\,iii) similarly requires its Fourier series to have
polynomially bounded (binary) rate of convergence;
and \cref{d:Two}b\,iii) requires its step function approximations
to have polynomially bounded (binary) rate of convergence.

\begin{definition}
\label{d:Moduli}
Fix $1\leq p\leq\infty$ and $f \in \ELL{p}(\IT)$. 
%Abbreviate $\barN:=\IN\cup\{-\infty\}$.
\begin{enumerate}
\item[a)] 
Generalizing \cref{d:Two}a\,iii),
define (two variants of) the binary \emph{Fourier rate}s 
$\varphi^{(p)}_f:\IN\nearrow\IN$
and $\varphi'^{(p)}_f:\IN\to\IN$
capturing convergence of $f$'s Fourier approximations, as
\begin{align*}
\varphi'^{(p)}_f(n)
\;:=\;&
\min
\big\{
m\in\N: \norm{f-\calF_{2^m}(f)}_{p}
\leq 2^{-n}
\big\} \quad\text{and}
\\
\varphi^{(p)}_f(n)
\;:=\;&
\min
\big\{
m\in\N: \exists g\in\calF_{2^m}. \norm{f-g}_{p}
\leq 2^{-n}
\big\} 
\enspace ;
\end{align*}
recall \cref{r:BestApprox} and see below about their subtle differences.
\item[b)]
%For \(S\in\N\), let \(\calS_S\) denote the space of complex functions on $[0;1)$ that are
%constant on each half-open interval $\big[s/S;(s+1)/S\big)$, $s=0,\ldots,S-1$.
%Assume that \(\norm{f}_{L^p}\leq 1\). 
Generalizing \cref{d:Two}b\,iii),
define the (binary) \emph{Step rate} $\sigma^{(p)}_f:\IN\nearrow\IN$,
namely of convergence of $f$'s approximations by piecewise averaged step functions, as
\begin{equation}
\label{eq:def-sigma}
\sigma^{(p)}_f(n)
\;:=\;
\min
\big\{
m\in\N:  \norm{f-\calS_{2^m}(f)}_p
\leq 2^{-n}
\big\} \enspace .
\end{equation}
\Martin{[Switched from \emph{best} approximation by equidistant step functions
to such with heights given by piecewise averages:
in order to agree with \cref{d:Two}b.]}
\item[c)]
Generalizing \cref{d:One},
define the (binary) $p$-\emph{modulus} 
$\mu^{(p)}_f:\IN\nearrow\IN$ by
\begin{equation}
\label{eq:def-mu}
\mu^{(p)}_f(n)
\;:=\;
\min
\big\{
m\in\N\;:\;
\forall |\delta|\leq2^{-m}. 
  \big\|f - \tau_\delta f\big\|_{p} \; \le \; 2^{-n}
\big\} \enspace .
\end{equation}
Here, $\big(\tau_\delta f\big)(t) := f(t+\delta \bmod 1)$ denotes the shift of $f$ by $\delta$ with wraparound.
\item[d)]
Generalizing (a)--(c),
fix a Banach space $\calB$ 
and a family of linear operators $\calA_L:\calB\to\calB$ ($L\in\IN$) 
such that $f=\lim\nolimits_L \calA_L f$ pointwise for each $f\in\calB$.
Define the (binary) rate $\alpha_f:\IN\nearrow\IN$ by
\[ \alpha_f(n) \;:=\;
\min\big\{\ell\in\IN: \forall L\geq2^\ell. \|f-\calA_L f\|\leq2^{-n}\big\} \enspace .
\]
\end{enumerate}
We may omit $f$ and/or $p$ when they are clear from context.
\end{definition}
In view of \cref{r:BestApprox}
note some subtleties of \cref{d:Moduli}a+b) in case $p\neq2$:
Unlike $\varphi^{(p)}_f$, 
$\varphi'^{(p)}_f\geq\varphi^{(p)}_f$ does generalize \cref{d:Two}a)
but may violate monotonicity.
Although $\calS_S(f)$ need not be the best $S$-step approximation to $f$ (\cref{r:BestApprox}), 
$\sigma^{(p)}_f$ still is non-decreasing:
since the averages on $2^{m}$ equidistant steps 
are precisely the pairwise averages on $2^{m+1}$ equidistant steps.
Fixing the number of steps to powers of two is arbitrary
and has (notable but) little effect on our conception of $\sigma^{(p)}_f$,
as demonstrated in \cref{sss:Step2Modulus} below.
%Note that $\mu^{(\infty)}$ extends the modulus of continuity 
%from \cref{d:One} to $\ELL{\infty}(\IT)\supset\calC(\IT)$.
The Norm Inequality~\eqref{e:Embeddings} implies
\begin{equation}
\label{e:Embeddings2}
\varphi^{(p)}_f\;\leq\;\varphi^{(q)}_f 
\quad\text{and}\quad
\sigma^{(p)}_f\;\leq\;\sigma^{(q)}_f
\quad\text{and}\quad
\mu^{(p)}_f\;\leq\;\mu^{(q)}_f 
\end{equation}
for $1\leq p\leq q\leq\infty$, when defined.
We record that 
$\mu^{(p)}_f(n)$ and 
$\sigma^{(p)}_f(n)$ and
$\varphi^{(p)}_f(n)$ are defined, namely finite,
for any $f\in\ELL{p}(\IT)$ with $1<p<\infty$.
%$\mu^{(p)}_f(n)$ and $\sigma^{(p)}_f(n)$ also for $p=1$.
Indeed, 
$\lim\nolimits_{h\to0} \tau_{h}f=f$ holds
for any $f\in\ELL{p}(\IT)$ 
according to \cite[Theorem~4.26]{Brez11FASS};
while $\mu_f^{(\infty)}$ is totally defined
iff $f$ is continuous, recall \cref{d:One}.
And equidistant step functions are dense in
\(\ELL{p}\): $\calS_S f\to f$ as $S\to\infty$;
moreover 
\(\calF_Kf\to f\) in \(\ELL{p}\) as \(K\to\infty\) holds by
the density of trigonometric polynomials and by the Marcel Riesz Theorem \cite{Forelli63}.
The latter asserts that, for fixed $1<p<\infty$, the operator norm
of $\calF_K:\ELL{p}\to\ELL{p}$ is bounded independently of $K$:
\begin{equation}
\label{eq:uniform-fourier-bound}
C_p
\;:=\;
\sup\nolimits_{K}
\Norm{\calF_K}_{p,p}
\;<\; \infty.
\end{equation}
The values of $C_p$ are known \cite{Pichorides1972} but 
shall not further concern our asymptotic considerations.
We write $\calO_p(1)$ for constants that may still depend on $p$.

\begin{remark}
\label{r:Unary}
Classical complexity theory measures algorithmic cost of,
say, integer factorization or primality testing in asymptotic dependence
on the binary length $n=\Log(N):=\lceil\log_2(N)\rceil$ of the input $N\in\IN$,
rather than of $N$ itself (=unary length).
Analogously \cref{d:Two}a\,ii) requires the $k$-Fourier coefficient 
be computable in time polynomial in $\log|k|$;
same for (b\,ii), as in \cref{d:One}ii).

And we consider `binary' error bounds $2^{-n}$ rather than `unary` $1/n$,
for instance regarding the modulus $\mu$ in \cref{d:One}
capturing quantitative continuity as parameterized by integer exponents to base $2^{\pm1}$:
following \cite[\S7.1]{Wei00} and \cite[\S2.4]{Ko91} and \cite[Ex~4.8.2]{Koh08a}.

A `unary' variant $M:\IN_+\nearrow\IN_+$ might instead require
$|f(t)-f(t')|\leq1/M(N)$ whenever $|t-t'|\leq1/N$.
%similarly for a unary counterpart to the binary $\ELL{p}$-modulus from \cref{d:LpModulus}.
In fact one can easily translate all asymptotic estimates in this Subsection
from binary to unary and back:
A unary modulus $M$ gives rise to binary $\mu:=\Log M(2^n)$,
and conversely binary $\mu$ gives rise to unary $M(N):=2^{\mu(\Log N)}$.

Our setting recovers 
traditional terminology capturing quantitative convergence
of a sequence $Y_K\to Y$ as follows,
here for the case of a Fourier series
$Y_K=\calF_{K}f$ with tail error $E^{(p)}_K(f)=\norm{f-\calF_Kf}_p$: \\[1ex]
\begin{tabular}{c@{\quad}c@{\qquad}c}
$\varphi_f^{(p)}(n)$ & $E^{(p)}_K(f)$ & common terminology \\ \hline
$n/s\pm o(n)$ & order $N^{-s}$ & \emph{algebraic} or \emph{polynomial} convergence \\
$2^{n/s}$ & $(\log N)^{-s}$ & logarithmic convergence \\
$\log_2(n)\pm\calO(1)$ & order $e^{-\calO(N)}$ &  \emph{exponential} or \emph{geometric} convergence \\
$\beta^{-1}\log_2(n)\pm\calO(1)$ & order $e^{-\calO(N^\beta)}$ &  \emph{stretched} exponential
%linear?
\end{tabular} \\[0.7ex]
These are correspondences between scales: 
asymptotic equivalence means corresponding upper as well as lower bounds.
Notions in literature vary \cite[\S9]{OrtegaRheinboldt70}.
\end{remark}
\Cref{t:Aras} establishes mutual linear bounds
among $\varphi$, $\sigma$, and $\mu$.
In particular the purely mathematical 
Conditions~(a\,iii) and (b\,iii) in \cref{d:Two}
are equivalent. 
By scaling, it suffices to treat functions of norm $\|f\|_p\leq1$:
\[ %\begin{equation} \label{e:Scaling}
\varphi^{(p)}_{f\cdot 2^r}(n)=\varphi^{(p)}_f(n+r), \qquad
\sigma^{(p)}_{f\cdot 2^r}(n)=\sigma^{(p)}_f(n+r), \qquad
\mu^{(p)}_{f\cdot 2^r}(n)=\mu^{(p)}_f(n+r) \enspace . 
\] %\end{equation}

\begin{theorem}
\label{t:Aras}
Fix \(1<p<\infty\). 
For every $f\in\ELL{p}(\IT)$ with $\|f\|_p\leq1$, it holds
\begin{enumerate}
\item[a)] $\displaystyle
\varphi'^{(p)}_f(n)\leq\; \quad\mu_f^{(p)}\big(n+\calO_p(1)\big)\:+\:\;n\:+\:\calO_p(1)$.
\item[b)] $\displaystyle
\mu_f^{(p)}(n) \;\leq\; \bm{2}\cdot\varphi_f^{(p)}(n+2)\qquad +\:\;n\:+\:\calO(1)$;
\\ $\displaystyle
\mu_f^{(p)}(n) \;\leq\; \bm{1}\cdot\varphi_f^{(p)}(n+2)\qquad +\:\;n\:+\:\calO(1)$ in case $p\geq2$.
\item[c)] $\displaystyle
\sigma_f^{(p)}(n)\;\leq\;\quad \mu_f^{(p)}(n+1)$.
\item[d)] $\displaystyle
\mu_f^{(p)}(n)\;\leq\;\quad \sigma_f^{(p)}(n+2)\qquad +\:\bm{p}n\:+\:\calO_p(1)$.
\item[e)] $\displaystyle
\sigma_f^{(p)}(n) \;\leq\;\quad \varphi_f^{(p)}(n+1)\qquad +\:\;n\:+\:\calO_p(1)$.
\end{enumerate}
$\calO_p(1)$ indicates constants
which may depend on $p$ but not on $n$ nor $f$.
\end{theorem}
Previous work \cite[\S2.4]{Aras25} would estimate $\sigma(n)\leq\textbf{2}\varphi\big(n+\calO(1)\big)+\calO(n)$
by combining (b)+(c) using $\mu$ as intermediate;
by comparing $\varphi$ and $\sigma$ directly,
we here obtain the tighter bound (e).
The technical proof is deferred to \cref{ss:Aras}.

\begin{corollary}
\label{c:Moduli}
For fixed $1<p<\infty$ and for all $f\in\ELL{p}(\IT)$ with $\|f\|_p\leq1$, it holds
\begin{gather*}
\calO_p\Big(\varphi^{(p)}_f\big(n\!+\!\calO_p(1)\big)+\!n\Big)
=
\calO_p\Big(\mu^{(p)}_f\big(n\!+\!\calO_p(1)\big)+\!n\Big)
=
\calO_p\Big(\sigma^{(p)}_f\big(n\!+\!\calO_p(1)\big)+\!n\Big)
\\[0.4ex]
\text{and}\qquad
\varphi^{(p)}_f\big(n+\calO_p(1)\big)+\calO_p(n)
\;=\;
\sigma^{(p)}_f\big(n+\calO_p(1)\big)+\calO_p(n) 
\enspace .
\end{gather*}
again with $\calO_p()$ quelling constant factors/offsets that may still depend on $p$.
\end{corollary}

%%%%%%%
\subsubsection[Sublinear Case]{Sublinear Case:}
\label{ss:Sublinear}

\cite[end of \S2.4]{Aras25} had erroneously claimed that the unary variants 
%(\cref{r:Unary})
of $\sigma$ and $\varphi$ be polynomially related.
This, however, is wrong due to the additive terms $+\calO_p(n)$ in \cref{t:Aras}.
In fact, linear equivalence breaks down in the \emph{sub}linear realm:

\begin{example}
\label{x:Moduli}
\begin{enumerate}
\item[a)]
Consider $f$ the rectangular unit pulse. 
It is already a step function in $\calS_2$,
hence has bounded Step rate $\sigma_f(n)\leq 1$
for any fixed $p\geq1$.
Its Fourier coefficients decay linearly:
$|\hat f_k|=\calO(1/k)$.
By Parseval, the Fourier series converges in $\ELL{2}$
asymptotically like (the square root of) the Basel series, that is,
with tail $\norm{f-\calF_{K}f}_{2}=\calO(\sqrt{K})$;
therefore $\varphi^{(2)}_f(n)=2n\pm\calO(1)$.
For sufficiently small $|\delta|$,
$\|f-\tau_\delta f\|_p=\sqrt[p]{2|\delta|}$
and $\mu^{(p)}_f(n)=pn+\calO(1)$. 
This illustrates the additive term $+pn$ in \cref{t:Aras}d).
\item[b)]
It is well-known \cite[Ch.~I, \S4, Exercise 4.4]{Katz04}
that a function $g:\IT\to\IC$ is smooth ~iff~ its Fourier coefficients decay
exponentially fast: $|\hat g_k|\leq 2^{-\calO(k)}$. 
The latter is in turn equivalent to Fourier rate 
$\varphi_g(n)\leq\log_2(n)\pm\calO(1)$;
recall \cref{r:Unary}.
Stretched exponential convergence $\varphi_g(n)\leq\calO(\log n)$
corresponds to smooth $g\in\calC^\infty[0;1)$,
with inverse $\beta$ of the constant factors $\beta^{-1}$ in $\calO(\log n)$ scaling Gevrey orders;
see \cite{ToftNabizadeh18,KMRZ15}.
\Martin{[One-to-one correspondence between $\beta^{-1}$ and Gevrey order,
or mutual bounds? for which $p$?]}
\item[c)]
For any fixed $p\geq1$, sublinear Step rate $n-\sigma(n)\to\infty$
requires a so-called \emph{regulated} function, 
namely a uniform limit of step functions \cite{RegulatedFunctions}.
Indeed, stretched exponential convergence here implies, for $n_k<n_{k+1}$ to be determined later:
\begin{eqnarray*}
%\label{e:Modulus1}
2^{1-n_k} &\geq&  
2^{-n_{k+1}}+2^{-n_k} \;\geq\; %&\geq&
\| g_{S_{k+1}} - g_{S_k} \|_p 
\\ &=&
\Big( \sum\nolimits_{s=0}^{S_{k+1}-1} | g_{S_{k+1},s} - g_{S_k,s} |^p \:/ S_{k+1}  \Big)^{1/p}
\\ &\geq& 
\sup\nolimits_{s=0}^{S_{k+1}-1} | g_{S_{k+1},s} - g_{S_k,s} | \:/ \sqrt[p]{S_{k+1}}
\\ &=&
\| g_{S_{k+1}} - g_{S_k} \|_\infty  \:/ 2^{\sigma^{(p)}_f(n_k)/p} \nonumber
\end{eqnarray*}
with $S_k:=2^{\sigma^{(p)}_f(n_k)}\leq S_{k+1}$
and $g_{S_k}\in\calS_{S_k}\subseteq\calS_{S_{k+1}}\subseteq\ELL{\infty}$
according to \cref{d:Moduli}b).
Now even in the worst-case $p=1$, 
$n-\sigma(n)\to\infty$ yields a sufficiently fast increasing integer sequence $n_k$ 
such that $\sum_{k>K} 2^{1-n_k+\sigma(n_k)/p}\to0$ as $K\to\infty$;
for which the above estimates show that
$g_{S_k}$ is a Cauchy sequence in $\ELL{\infty}$.
\Martin{[Please double-check! How about a converse?]}
\item[d)]
Even for $p=1$, sublinear $p$-modulus $n-\mu(n)\to\infty$
requires the underlying function $f\in\ELL{p}(\IT)$ to be constant almost everywhere.
To see this, record that $|\exp(it)-1|\geq |t|/2$ whenever $|t|\leq\pi$.
Therefore 
\[
|\hat f_k| \cdot \pi|k||\delta|
\;\leq
|\hat f_k| \cdot \big| 1\:-\:\exp(2\pi i k\delta)\big|
\;=
\big|\widehat{\big(f-\tau_{\delta} f\big)_k}\big| 
\;\leq
\| f-\tau_\delta f\|_1
%\;\leq\; \| f-\tau_\delta f\|_p
\;\leq2^{-n}
\]
whenever $|\delta|\leq2^{-\mu(n)}$ and $|\delta|\leq1/|k|$.
For $k\neq0$, choosing $\delta:=2^{-\mu(n)}/|k|$ yields
$|\hat f_k|\leq 2^{-n+\mu(n)}/\pi\to0$ as $n\to\infty$ by hypothesis,
so all Fourier coefficients vanish except for $\hat f_0$;
by uniqueness, $f$ must be constant almost everywhere.
\end{enumerate}
\end{example}

%%%%%%%
\subsubsection[Exponential Case]{Exponential Case:}
\label{ss:Exponential}

According to \cref{e:Embeddings},
the modulus of continuity is an upper bound
to the $\ELL{p}$-modulus. 
H\"{o}lder-continuous functions 
(\cref{f:Modulus}c)
therefore also have linear moduli/rates
when considered embedded in $\ELL{p}$.
Similarly to \cref{f:Modulus}d),
the present subsection constructs integrable functions 
with exponentially large moduli/rates.
By \cref{t:Aras}, a function with asymptotically large $\ELL{p}$--modulus
also has large Step/Fourier rate, and vice versa.

\begin{example}%[Lacunary Fourier Series]
\label{x:Lacunary}
Consider 
$\displaystyle
f(t)\;:=\;\sum\nolimits_{k\geq0} \exp\bigl(2\pi i t\,2^k\bigr)/{2^k} \;\in\;\ELLtwo(\IT)$.
\\ The lacunary series forces its binary
Hilbert Fourier rate
$\varphi_f^{(2)}(n)$ to grow exponentially in $n$. Similarly, 
$\sum\nolimits_{K\geq0} \exp\bigl(2\pi i t\,2^{2^k}\bigr)/{2^{2^k}} \;\in\;\ELLtwo(\IT)$
has even doubly exponential binary Fourier rate.
\end{example}
Note that, like Liouville's historically first construction of transcendental numbers,
our example be admittedly be considered `artificial'.
The following construction, based on \cref{f:Modulus}d), is arguably less artificial:

\begin{example}
\label{x:Exponential}
The derivative $h'$ on $(0;1]$ of $h(t)=1/\ln(e/t)$ from \cref{f:Modulus}d)
belongs to $\ELL{1}(\IT)$, but not to $\calC(\IT)$.
By the \emph{Fundamental Theorem of Calculus},
its $\ELL{1}$-modulus $\mu^{(1)}_{h'}$ is essentially 
the modulus of continuity $\mu^{(\infty)}_h$ of $h$,
and therefore also grows exponentially in $n$.

Similarly, $(h\circ h)'$ belongs to $\ELL{1}(\IT)$
and has doubly exponential $\ELL{1}$-modulus.
\end{example}
\Cref{x:ONB} demonstrates that, over the Hilbert space $\ELLtwo$,
proceeding from the Fourier $\exp(2\pi ikt)$ orthonormal basis 
to a different orthonormal basis may drastically affect the asymptotic rate of approximation
in the sense of \cref{d:Moduli}d):

\begin{example}
\label{x:ONB} 
Consider $\calB:=\ELLtwo$ and some orthonormal basis $(e_0,e_1,\ldots,e_\ell,\ldots)$ 
with induced projections 
\[ 
\calA_L:\sum\nolimits_{\ell=0}^\infty c_\ell e_\ell \;\mapsto\; \sum\nolimits_{\ell\leq L} c_\ell e_\ell ,
\qquad L\in\IN
\]
and induced rates $\alpha_f:\IN\nearrow\IN$ according to \cref{d:Moduli}d).
Taking $e_\ell:=\exp(2\pi i\ell t)\in\ELLtwo$ recovers $\alpha=\varphi^{(2)}$.

But now take instead any $f\in\ELLtwo$ with asymptotically `large' 
$\ELLtwo$-modulus $\mu^{(2)}_f$,
such as from \cref{x:Exponential}; recall \cref{e:Embeddings2}.
Normalize, and extend the single 
vector $f_0:=f/\|f\|_2$ to some orthonormal basis $(f_\ell)$ of $\ELLtwo$.
By \cref{t:Aras}, $f_0$ has also asymptotically `large' Fourier rate $\varphi^{(2)}(n)$.
But in the orthonormal basis $(f_\ell)$, $f$ has coefficient sequence $(\|f\|_2,0,0,\ldots)$:
constant rate $\alpha_f\leq 1$.
\end{example}

%%%%%%%%%%%%%%%%%%%%%%%%%%%%%%%%%%%%%%%%%
\subsection{Linearly Equivalent Moduli: Proof of \Cref{t:Aras}}
\label{ss:Aras}

This subsection uses Approximation Theory and Hard Analysis to prove \cref{t:Aras}.
First we record some folklore estimates among $p$-norms 
of vectors, of sequences, and of functions where $p\geq q\geq1$
and $p':=\tfrac{p}{p-1}$:
\begin{gather}
\label{e:Hoelder0}
(|x|+|y|)^p \;\leq\; 2^{p-1} (|x|^p+|y|^p)  
%\enspace ,
\\
\label{e:Hoelder1}
\sum\nolimits_n |x_n\cdot y_n|
\;\leq\; 
\Big( \sum\nolimits_n |x_n|^p \Big)^{1/p} \:\cdot\:
\Big( \sum\nolimits_n |y_n|^{p'}  \Big)^{1/p'} 
%\enspace ,
\\
\label{e:Embeddings3} 
\Big(\sum\limits_{k=0}^{K-1} |x_k|^p\Big)^{1/p}
\;\;\leq\;\;
\Big(\sum\limits_{k=0}^{K-1} |x_k|^q\Big)^{1/q}
\;\;\leq\;\;
K^{\tfrac{1}{q}-\tfrac{1}{p}}\cdot
\Big(\sum\limits_{k=0}^{K-1} |x_k|^p\Big)^{1/p}
%\enspace ,
\\
\label{e:Hoelder2} 
\int\nolimits_X |f(x)\cdot g(x)|\,dx \;\leq\; 
\Big( \int\nolimits_X |f(x)|^p \,dx \Big)^{1/p} \:\cdot\:
\Big( \int\nolimits_X |g(x)|^{p'} \,dx \Big)^{1/p'} 
%\enspace ,
\\
\label{e:Minkowski}
\Big( \int\nolimits_Y \Big| \int\nolimits_X F(x,y)\, dx\Big| ^{p} \, dy\Big)^{1/p}
\;\leq\; 
\int\nolimits_X \Big( \int\nolimits_Y |F(x,y)|^p \, dy \Big) ^{1/p} \, dx
%\enspace , 
\\
\label{e:Jensen} 
\gamma\big(\int\nolimits_\Omega f\,d\mu\big) \;\leq\;  \int\nolimits_\Omega \gamma\circ f\,du
\end{gather}
when $\gamma:\IR\to\IR$ is convex and $\mu$ a probability measure on $\Omega$
with measurable $f:\Omega\to\IR$.
Inequality~\eqref{e:Jensen} is Jensen's, and also \eqref{e:Hoelder0} follows from convexity;
\eqref{e:Hoelder1} and \eqref{e:Hoelder2} are H\"{o}lder's Inequalities, 
which imply \eqref{e:Embeddings3} and \eqref{e:Embeddings2} respectively.
Finally \eqref{e:Minkowski} is Minkowski's Integral Inequality.

In the sequel,
the restriction \(1<p<\infty\) enters through the uniform boundedness of the
Fourier partial-sum operators \(\calF_N\) on \(\ELL{p}(\IT)\). At the endpoint \(p=1\),
the norms of the Dirichlet partial-sum operators are not uniformly bounded, so
the direct argument above does not apply without additional assumptions.
It suffices to bound moduli/rates of approximations:

\begin{lemma}
\label{l:ApproxRate}
Suppose $\calA_L$ in \cref{d:Moduli}d)
have bounded operator norm $\Norm{\calA_L}\leq2^c-1$ 
for some $c\in\IN$ independently of $L$.
If $\|f-f_m\|\leq2^{-m}$ and $f_m$ have rates $\alpha_m$,
then $f$ has rate $\alpha_f(n)\leq\alpha_{n+c+1}(n+1)$. 
\end{lemma}
\begin{proof}
$\|f-\calA_L\| 
\;\leq\;
\|f-f_m\|
\:+\:
\|f_m-\calA_L f_m\|
\:+\:
\|\calA_L(f_m-f)\|
\;\leq\;
2^{-m} \:+\: 2^{-n-1} \:+\: (2^c-1)\cdot 2^{-m} 
\;\leq\; 2^{-n}$
for $L=2^{\alpha_m(n+1)}$ and $m:=n+c+1$.
\qed\end{proof}

\renewcommand*{\thesubsubsection}{\thesubsection\alph{subsubsection}}
%%%%%%%
\subsubsection[From Modulus to Fourier Rate]{From Modulus to Fourier Rate:}
\label{sss:Modulus2Fourier}

Recall the Ces\`{a}ro average 
$\sum_{k=0}^{K} \calF_k f/(K+1)=f\ast F_K\in\calF_K$ 
of $f$'s first $(K+1)$ Fourier sums
$\calF_K f= f\ast D_K\in\calF_K$,
both convolutions with respectively the Dirichlet and the Fej\'{e}r kernel:
\begin{eqnarray*}
D_K(t) &:=&\sum\nolimits_{k=-K}^{+K} \exp(2\pi i k t) \;=\; \frac{\sin\big(\pi (2K+1) t\big)}{\sin(\pi t)} \enspace , \\
F_{K}(t) &:=& \sum\nolimits_{k=0}^{K} D_k(t)/(K+1) \;=\; 
\Big|\sum\nolimits_{k=0}^K \exp(2\pi i kt)\Big|^2/(K+1) \\
&& \qquad\qquad\qquad\qquad\qquad=
\Big( \frac{\sin\big(\pi (2K+1) t\big)}{\sin(\pi t)} \Big)^2 / (K+1)  \enspace . 
\end{eqnarray*}
Record $F_K(t)\geq0$ and $F_K(-t)=F_K(t)$ and $\int F_K(t)\,dt=1$.
The convolution
\[ f\ast F_K \;=\; \int\nolimits_0^1 F_K(h)\cdot \tau_{h} f\,dh \;=:\; T_K \]
is a trigonometric polynomial of degree $\leq K$, hence $\calF_K T_K=T_K$
(but $T_K\neq\calF_K f$ for $p\neq2$, recall \cref{r:BestApprox}).
Following the proof of \cite[Theorem~1.10]{Duoandikoetxea}, 
\begin{eqnarray*}
\| T_K-f \|_p 
&=&
\bigg( \int\nolimits_0^1 
\Big| \int\nolimits_{-1/2}^{+1/2} \big( f(t+h)-f(t) \big)\cdot F_K(h)\,dh \Big|^p 
\,dt\bigg) ^{1/p} \\
&\overset{\eqref{e:Minkowski}}{\leq}&
\int\nolimits_{-1/2}^{+1/2} 
\Big( \int\nolimits_0^1 |f(t+h)-f(t)|^p \cdot F_K^p(h)\, dt
\Big)^{1/p} \, dh \\
&=& 
\int\nolimits_{-1/2}^{+1/2}
\| \tau_{h} f -f \|_p \cdot F_K (h)\, dh \\
&\leq&
\int\nolimits_{|h|<\delta}
\underbrace{\| \tau_{h} f -f \|_p}_{\leq2^{-(n+1)}} 
\cdot F_K (h)\, dh \;+\;
\int_{|h|\geq\delta}^{1/2} 
\underbrace{\big(\|\tau_h f\|_p+\|f\|_p\big)}_{\leq2}\cdot F_K(h)\,dh  \\
&\leq& 2^{-(n+1)} \;+\; \calO\big(2^{\mu(n+1)}/(K+1)\big) 
\quad\leq 2^{-n}
\end{eqnarray*}
for $\delta:=2^{-\mu(n+1)}$ and $K\geq 2^{\mu(n+1)+n+\calO(1)}$
since 
$\int\nolimits_{\delta}^{1/2} F_K(h)\,dh \;=$
\[ =\; %&\leq& 
\int\limits_{\delta}^{1/2} \sin^{-2}(\pi h) \,dh / (K+1)
\:=
\frac{-\cos(\pi h)}{\pi\sin(\pi h)\cdot(K+1)} \Big|_{h=\delta}^{h=1/2}
\;\leq \calO(1/\delta)/(K+1). \] 
It follows that
\[ 
\| \calF_K f- f \|_p
\;\leq\;
\| \calF_K (f- T_K) \|_p \;+\;
\| T_K-f \|_p 
\;\leq\; (C_p+1)\cdot 2^{-n} 
\]
for the operator norm $\Norm{\calF_K}_{p,p}\leq C_p<\infty$ from \cref{eq:uniform-fourier-bound}.
Together with \cref{l:ApproxRate}, this proves \cref{t:Aras}a).
Note that, although we use the Fej\'er Mean for auxiliary approximation estimates, 
the final approximant is the ordinary Fourier partial sum \(\calF_Kf\):
hence applies to $\varphi'^{(p)}_f$ and not just $\varphi^{(p)}_f$,
recall \cref{r:BestApprox}.

%%%%%%%
\subsubsection[From Fourier Rate to Modulus]{From Fourier Rate to Modulus:}
\label{sss:Fourier2Modulus}

First record that
$|\exp(ir)-1|\leq |r|$ holds for all $r\in\IR$.
Next estimate the $\ELL{p}$-modulus of a trigonometric polynomial
$T_K(t)=\sum\nolimits_{k=-K}^{+K} c_k \cdot\exp(2\pi ikt)\in\calF_K$:

%$\|T_K-\tau_\delta T_K\|_p^p$
\begin{multline*}
\|T_K-\tau_\delta T_K\|_p^p 
\;=\;
\int\nolimits_0^1 
\Big| \sum\nolimits_{k=-K}^{+K} \exp(2\pi i k t) \cdot c_k \cdot  \big(1-\exp(2\pi i k \delta)\big) \Big|^p \, dt \\
\overset{\eqref{e:Hoelder1}}{\leq}\;
\int\limits_0^1 
\Big( \underbrace{\sum\nolimits_{k=-K}^{+K} |\exp(2\pi i k t)|^{p'}}_{=2K+1} \Big)^{\tfrac{p}{p'}} 
\cdot 
\Big( \sum\limits_{k=-K}^{+K} |c_k|^p\cdot \underbrace{\big|1-\exp(2\pi i k \delta)\big|^p}_{\leq |2\pi k\delta|^p} \Big) \, dt \\
\leq\;
\big(2K+1\big)^{\tfrac{p}{p'}=p-1} \cdot \big| 2\pi K \delta\big|^p \cdot \|\vec c\|_p^p
\end{multline*}
with the Fourier coefficient vector $\vec c\in\IC^{2K+1}$. \\
Then 
$|c_k| 
%\;=\;\Big| \int\nolimits_0^1 T_K(t)\cdot\exp(-2\pi ikt)\,dt \Big| 
\leq
%\int\nolimits_0^1 |g_K(t)|\,dt \;=\; 
\|T_K\|_1 \leq \|T_K\|_p$
by \cref{e:Embeddings2},
and therefore 
\begin{equation}
\label{e:Estimate}
\|\vec c\|_p 
\;=\; \Big( \sum\nolimits_{k=-K}^{+K} |c_k|^p \Big)^{1/p} 
\;\leq\; \big(2K+1\big)^{1/p} \cdot \|T_K\|_p 
\enspace . \end{equation}
In conclusion,
$
\|T_K-\tau_\delta T_K\|_p \leq
(2K+1) \cdot | 2\pi K \delta| \cdot \|T_K\|_p
\leq2^{-n-1}$
whenever $|\delta|\leq 2^{-\mu_K(n)}$
for $\mu_K(n):=n+\calO(1)+\bm{2}\Log(K)$,
provided $\|T_K\|_p\leq2$.
With $K\leq2^{\varphi(n+2)}$ and $T_K\in\calF_K$ such that
$\norm{f-T_K}_p\leq2^{-n-2}$,
$\|f-\tau_\delta g\|_p \;\leq\;
\|f-T_K\|_p 
\:+\:
\|T_K-\tau_\delta T_K\|_p 
\:+\:
\|\tau_\delta (T_K-f)\|_p 
\leq2^{-n-2}\:+\:2^{-n-1}\:+\:2^{-n-2}\;=\;2^{-n}$
since $\tau_\delta$ is an isometry;
cmp. \cref{l:ApproxRate}.
This proves the first line of \cref{t:Aras}b).
In case $p\geq2$ one can improve Estimate~\eqref{e:Estimate}
for the second line of \cref{t:Aras}b)
%from $\bm{2}\varphi(\cdots)$ to $\bm{1}\varphi(\cdots)$,
with \emph{Plancherel}:
\[ \|\vec c\|_p 
\;\overset{\eqref{e:Embeddings3}}{\leq}\;
\|\vec c\|_2 
\;=\; 
\|T_K\|_2 
\;\overset{\eqref{e:Embeddings2}}{\leq}\;
\|T_K\|_p 
\enspace . \]

%%%%%%%
\subsubsection[From Modulus to Step Rate]{From Modulus to Step Rate:}
\label{sss:Modulus2Step}

To $f\in\ELL{p}(\IT)$ and $S\in\IN$ consider associated step functions 
$f_S\in\calS_S(\IT)$ according to \cref{e:StepAverage}.
For \(x\in [s/S;(s+1)/S)=:\Omega\), 
$\gamma(z):=|z|^p$ %and $\Omega:=\big[s/S;(s+1)/S\big)$
and $\mu:=S\cdot\lambda$
in Jensen's inequality~\eqref{e:Jensen} gives 
\begin{multline}
\label{e:StepEstimate}
\big| 
f(x)
-
\calS_S(f,x)\big|^p
\;=\;
\abs{
S\cdot\int\nolimits_{s/S}^{(s+1)/S}
\big(f(x)-f(y)\big)\,dy
}^p \\
\leq\;
S\int\nolimits_{s/S}^{(s+1)/S} 
\abs{f(x)-f(y)}^p\,dy \enspace .
\end{multline}
Using $\bigcup_{s=0}^{S-1} \big[s/S;(s+1)/S\big)^2 
\subseteq \big\{(x,y)\in\IT^2: |x-y|\leq1/S\big\}$
in (*) yields
\begin{eqnarray}
\norm{f-\calS_S(f)}_p^p 
&=& \nonumber 
\sum\nolimits_{s=0}^{S-1}
\int\nolimits_{s/S}^{(s+1)/S} 
\big|f(x)-\calS_S(f,x)\big|^p\,dx
\\ &\overset{\eqref{e:StepEstimate}}{\leq}& \nonumber
\sum\nolimits_{s=0}^{S-1} 
\int\nolimits_{s/S}^{(s+1)/S}
S\cdot\int\nolimits_{s/S}^{(s+1)/S}
|f(x)-f(y)|^p\,dx\,dy 
\\ &\overset{(*)}{\leq}& \nonumber
S\cdot\int\nolimits_{|x-y|\leq1/S}
|f(x)-f(y)|^p\,dx\,dy
\\ &=& 
S\cdot
\int\nolimits_{-1/S}^{1/S}
\norm{\tau_h f -  f}_{p}^p\,dh 
\label{e:Modulus2Step}
\end{eqnarray}
$\leq2\sup_{|h|\leq1/S} \big\| \tau_h f-f \big\|_p^p 
\leq2^{-pn}$ for $S:=2^{\mu^{(p)}_f(n+1)}$.
This proves \cref{t:Aras}c).

%%%%%%%
\subsubsection[From Step Rate to Modulus]{From Step Rate to Modulus:}
\label{sss:Step2Modulus}

\Cref{d:Moduli}b) considers approximations by step functions 
consisting of exactly $S=2^m$ equidistant steps, starting at the origin,
with heights given by local averages.
For later use, such as in \cref{ss:Average}, we also consider
$S\leq2^m$ steps which may be only approximately equidistant and start anywhere
and have heights that yield the best $p$-approximation:

\begin{definition}
\label{d:Step2}
Fix finite $\vec T\subseteq[0;1)$, $\vec T=\{t_0<\ldots<t_{S-1}\}$.
Let $\calS_{\vec T}$ denote the space of functions $g:\IT\to\IC$ that are
constant on each half-open interval $\big[t_s;t_{s+1}\big)$, $s=0,\ldots,S-2$
as well as on the `wrapped interval' $[t_{S-1};1)\cup[0;t_0)$. 
Similar to \cref{e:StepAverage}, write 
\begin{equation}
\label{e:StepAverage2}
\calS_{\vec T}(f,t)\;:=\;\bar f_{\vec T,s}, \quad t\in\big[t_s;t_{s+1}\big) %, 0\leq s<S-1
\text{ or } t\in[t_{S-1};1)\cup[0;t_0) %\enspace , %s=S-1 \enspace ,
\end{equation}
with the local averages 
$\bar f_{\vec T,s}:=\int_{t_s}^{t_{s+1}} f(t)\,dt/(t_{s+1}-t_s)$
and $\bar f_{\vec T,S-1}:=
\big(\int_{t_{S-1}}^{1} f(t)\,dt + \int_{0}^{t_0} f(t)\,dt\big)/(1-t_{S-1}+t_0)$.
For $0<\varepsilon\leq1$, $\vec T$ is called $\varepsilon$-\emph{equidistant} if it holds
\begin{equation}
\label{e:Equidist}
\frac{\varepsilon}{S} \;\leq\; t_{s+1}-t_s \;\leq\; \frac{1}{\varepsilon S}
\qquad\text{and}\qquad
\frac{\varepsilon}{S} \;\leq\; 1- t_{S-1} + t_0 \;\leq\; \frac{1}{\varepsilon S} \enspace .
\end{equation}
Abbreviate 
$\calT_{S,\varepsilon}:=\big\{ \vec T \text{ is } \varepsilon\text{-equidistant}, S=\Card\vec T\big\}$
and $\calS_{S,\varepsilon}:=\bigcup\big\{\calS_{\vec T}: \vec T\in\calT_{S,\varepsilon}\big\}$;
hence $\calT_{S,1}=\big\{\{0,\tfrac{1}{S},\tfrac{2}{S},\ldots,\tfrac{S-1}{S}\}\big\}$
and $\calS_S=\calS_{S,1}$. Finally let
\begin{equation}
\label{e:sigma2}
\sigma^{(p)}_{f,\varepsilon}(n)
\;:=\;
\min
\big\{
m\in\N: \exists S\leq2^m. \exists g\in \calS_{S,\varepsilon}. 
\norm{f-g}_{p} \leq 2^{-n}
\big\} \enspace .
\end{equation}
\end{definition}
Note that, for $S\leq2^m$, 
$\{0,\tfrac{1}{S},\tfrac{2}{S},\ldots,\tfrac{S-1}{S}\}$ 
can be extended to some $\vec T\in\calT_{2^m,1/2}$;
hence $\calS_S\subseteq\calS_{2^m,1/2}$.
Obviously $\sigma^{(p)}_{f,\varepsilon}\leq\sigma^{(p)}_f$.
In order to establish \cref{t:Aras}d), we now show the stronger inequality
\begin{equation}
\label{e:Step2Modulus}
\mu^{(p)}_{f}(n) \;\leq\; \sigma^{(p)}_{f,\varepsilon}(n+2)\:+\:pn\:+\:\calO_{p}(1)\:+\:\calO_{\varepsilon}(1)\enspace .
\end{equation}
For $0<\delta\leq\varepsilon/S$,
estimate the $\ELL{p}$-modulus of an $\varepsilon$-equidistant $\vec T$-step function
\begin{equation}
\label{e:StepFunc}
g \;=\;\sum\nolimits_{s=0}^{S-2} 
h_s \cdot\chi_{[t_s;t_{s+1})} \;+\; h_{S-1}\cdot \chi_{[t_{S-1};1)\cup[0;t_0)} 
\enspace : 
\end{equation}
%of norm $\norm{g}_p^p\;=\;\sum\nolimits_{s=0}^{S-2} |c_s|^p\cdot(t_{s+1}-t_s) + |c_{S-1}|^p\cdot(1-t_{S-1}+t_0)$
\begin{align*}
\| g - \tau_\delta g \|^p_p
&=
\sum\nolimits_{s=0}^{S-2} \delta \cdot\big|h_{s+1}-h_s\big|^p 
\;+\; \delta\cdot\big|h_{S-1}-h_0\big|^p \\
&\overset{\eqref{e:Hoelder0}}{\leq}
\delta 2^{p-1} \cdot \sum\nolimits_{s=0}^{S-2} \big(|h_{s+1}|^p+|h_s|^p\big)
\;+\;\delta 2^{p-1} \big(|h_{S-1}|^p+|h_0|^p\big) \\
&\overset{\eqref{e:Equidist}}{\leq}
\delta 2^p \cdot\sum\nolimits_{s=0}^{S-2} |h_s|^p\cdot (t_{s+1}-t_s)\cdot S/\varepsilon
\\ & \quad \;+ \delta 2^p \cdot |h_{S-1}|^p\cdot (1-t_{S-1}+t_0)\cdot S/\varepsilon  \\[0.3ex]
&= S \delta 2^{p} \|g\|_p^p  /\varepsilon
\end{align*}
$\leq 2^{2p} S \delta/\varepsilon$
provided that $\|g\|_p\leq2$.
Now $2^{2p} S \delta/\varepsilon{\leq} 2^{-(n+1)p}$
for $\delta\leq 2^{-\mu(n)}$
where $\mu(n):=pn+3p+\Log(S)+\Log(1/\varepsilon)$.
Applied to $g\in\calS_{S,\varepsilon}$ 
with $\norm{f-g}_p\leq2^{-n-2}$
for $S\leq2^m$ and $m:=\sigma^{(p)}_{f,\varepsilon}(n+2)$ 
gives 
$\|f-\tau_\delta f\|_p
\leq
\|f-g\|_p
+
\|g-\tau_\delta g\|_p
+
\|\tau_\delta(g-f)\|_p
\leq
2^{-n-2}+2^{-n-1}+2^{-n-2}$
since $\tau_\delta$ is an isometry;
cmp. \cref{l:ApproxRate}.
This proves \cref{t:Aras}d), and slightly more.

%%%%%%%
\subsubsection[From Fourier to Step Rate]{From Fourier to Step Rate:}
\label{sss:Fourier2Step}

Extending \cref{e:Fourier} from $\calK=\{-K,\ldots,+K\}$,
consider any finite subset $\calK\subseteq\IZ$ and let
$\calF_\calK$ denote the vector space of trigonometric polynomials
$\sum_k c_k\cdot\exp(2\pi ikt)$
whose coefficients vanish $c_k=0$ for all $k\in\IZ\setminus\calK$;
$\calF_\calK(f,t):=\sum\nolimits_{k\in\calK} \hat f_k\cdot\exp(2\pi i kt)$
for the Fourier projection onto the remaining frequencies, \ie those in $\calK$.

We first show that a continuously differentiable function $g\in\calC^1(\IT)$
can be approximated well by equidistant step functions.
To this end estimate
\begin{eqnarray}
\norm{g-\calS_S(g)}_p^p 
&\overset{\eqref{e:Modulus2Step}}{\leq}& \nonumber 
S\cdot
\int\nolimits_{-1/S}^{1/S}
\norm{\tau_h g -  g}_{p}^p\,dh 
\\ &\leq& 
S\cdot \norm{g'}_{p}^p \cdot
\int\nolimits_{-1/S}^{1/S}
\abs{h}^p\,dh 
\;=\; \tfrac{2}{p+1} \cdot\norm{g'}_p^p / S^p
\label{e:Fourier2Step2}
\end{eqnarray}
for continuously differentiable $g$,
since
$\tau_hg-g = \int_0^h\tau_\delta g'\,d\delta$ 
and thus
\begin{eqnarray*}
\norm{\tau_h g - g}_{p}
&=&
\Big( \int\nolimits_{t=0}^1 \Big|
\int\nolimits_0^h g'(t+\delta)\,d\delta\Big|^p\,dt\Big)^{1/p} \\
&\overset{\eqref{e:Minkowski}}{\leq}&
\int\nolimits_0^h \Big(\int\nolimits_0^1 |g'(t+\delta)|^p\,dt\Big)^{1/p}\,d\delta 
\;= \int\nolimits_0^h \|\tau_\delta g'\|_p\,d\delta
\;= |h|\cdot\|g'\|_p \enspace .
\end{eqnarray*}
We next derive the $\ELL{p}$ Bernstein-type polynomial inequality 
\begin{equation}
\label{eq:bernstein}
\norm{T'}_{p}
\;\leq\;
\calO_p(1)\cdot\deg(T)\cdot\norm{T}_{p}
\end{equation}
for trigonometric polynomials 
$T(t):=\sum\nolimits_{\abs{k}\leq K}
c_k\cdot\exp(2\pi i kt)$
of degree $\deg(T)\leq K$.
To this end note that $T=\calF_K(T)$ has derivative
\[ %\begin{equation}\label{eq:derivative-decomposition}
T'
\;=\;
2\pi i\cdot\sum\nolimits_{k=1}^{K}
\calF_{[k;K]}T \;-\; 
2\pi i\cdot\sum\nolimits_{k=1}^{K} \calF_{[-K;-k]}T
\enspace .
\] %\end{equation}
Indeed, any positive Fourier mode \(k>0\) appears in the first sum exactly
\(k\) times, while a negative mode \(k<0\) appears in the second sum
exactly \(|k|=-k\) times and is preceded by a minus sign. 
Consequently,
\[
\norm{T'}_{p}
\;\leq\;
2\pi\sum\nolimits_{k=1}^{K}
\Norm{\calF_{[k;K]}}_{p,p}\cdot\norm{T}_{p}
\;+\;
2\pi\sum\nolimits_{k=1}^{K}
\Norm{\calF_{[-K;-k]}}_{p,p}\cdot\norm{T}_{p}
\enspace .
\]
Similarly to Marcel Riesz' \cref{eq:uniform-fourier-bound},
the projections $\calF_\calK:\ELL{p}(\IT)\to\ELL{p}(\IT)$ 
have operator norm $\Norm{\calF_\calK}_{p,p}\leq\calO_p(1)$
\ie bounded independently of $\calK\subseteq\IZ$
for any fixed $p\in(1,\infty)$;
compare \cite[Theorem~4.10]{BanachSpaceTheory11}.
\Martin{[Should we explicitly elaborate, like this:]}
Explicitly, in case interval $\calK:=[k,K]\cap\IZ$ or $\calK:=[-K,-k]\cap\IZ$ has odd cardinality,
then \(\calK=\ell+[-L,L]\) for some integers \(\ell\) and \(L\geq 0\). 
Hence $\calF_\calK=\calM_\ell\circ \calF_L \circ \calM_{-\ell}$
where $\calM_\ell f(t):=\exp(2\pi i\ell t)\cdot f(t)$
denotes modulation which is an isometry on \(L^p(\IT)\);
thus $\Norm{\calF_\calK}_{p,p}\leq C_p$ in this case.
In case non-empty interval $\calK\subseteq\IZ$ has even cardinality,
interval $\calK':=\calK\setminus\{\min\calK\}$ has odd cardinality
and $\calF_\calK=\calF_{\calK'}+\calF_{\{\min\calK\}}$
where $\calF_{\{k\}}f(t)=\hat f_k\cdot\exp(2\pi ikt)$
has $\Norm{\calF_{\{k\}}}_{p,p}=1$. \Martin{[Really $=1$, independent of $k$ and $p$?]}
This proves \cref{eq:bernstein}.

Now apply Equations~\eqref{e:Fourier2Step2} and \eqref{eq:bernstein} together
to $T:=g:=\calF_K(f)$ for $K:=2^{\varphi_f^{(p)}(n+1)}$,
noting $\|T\|_p \leq \|T-f\|_p+\|f\|_p\leq 2$:
\[ 
\norm{f-\calS_S(T)}_p
\;\leq\;
\norm{f-T}_p \;+\; \norm{T-\calS_S(T)}_p
\;\leq\;
2^{-n-1} \;+ \calO_p(1)\cdot K/S 
\;\leq2^{-n}
\]
for $S:=2^{\sigma(n)}$ where
$\sigma(n):=\varphi_f^{(p)}(n+1)+n+\calO_p(1)$.
This proves \cref{t:Aras}e).

%%%%%%%%%%%%%%%%%%%%%%%%%%%%%%%%%%%%%%%%%
\subsection{about Step-Computability on Average, \cref{d:Two}c}
\label{ss:Average}

\Cref{d:Two}c) may perhaps seem surprising and subtle,
so let us elaborate and clarify.
Note that among \cref{d:Two}b\,iii) and \cref{d:Two}c\,i) and \cref{e:Average},
any two imply the third where 
\begin{equation}
\label{e:Average}
\norm{f-g_n}_2\leq2^{-n} \quad\text{for}\quad S(n)=2^{\poly(n)}
\enspace , \quad
g_n  :=\sum\limits_{s=0}^{S-1} 
h_{s,n} \cdot\chi_{[s/S;(s+1)/S)}
\end{equation} 
Both \cref{d:Two}b\,i+ii) and \cref{d:Two}c\,i+ii) require that 
the equidistant averaged $S$-step approximation $\calS_S(f)$
be computed up to error $\leq 2^{-n}$
within runtime polynomial in $n$;
but in (b) for every $S\in\{ 2^{P(m)}:m\in\IN\}$ for some polynomial $P\in\IN[M]$ 
with $n$ independent of $m$, while in (c) only for $m=n$.
Put differently, (b\,i+ii) requires computationally approximating 
up to error $\leq2^{-n}$ each local average $\bar f_{S,s}$ individually (\ie w.r.t maximum norm),
whereas (c\,i+ii) amounts to approximation with residual sum of squares error $\leq2^{-n}$;
similarly for (b\myprime) and (c\myprime) within complexity \classSharpP instead of \classFP.
In particular (b) implies polynomial-time computability 
of the global average $\bar f_{1,0}$ whereas (c) does not:
a difference we shall harness to prove polynomial non-equivalence,
such as in \cref{t:StepMeanNotOther}.

In the \ELL{p} setting, averaged $p$-th power error bounds (c\,i) 
are arguably more natural than the maximum norm underlying (b\,i).
%\Cref{t:SharpP,t:StepMeanNotOther} below
%apply to (c) and (c\myprime) rather than to (b) and (b\myprime), respectively.
\emph{Leonid Levin}'s notion of average-case complexity
\cite[Definition~18.4]{Arora09} 
actually arises frequently in Real Computation \cite{SSZ16x,KTZ18x}. 
Average polynomial runtime in dependence on the worst-case precision
is equivalent to average polynomial precision in dependence
on the worst-case runtime:

\begin{observation}
\label{o:Average}
Let $\calA$ denote an algorithm producing on input $s\in\IN$
an infinite sequence of Gaussian rational approximations $\big(h_{s,m}\big)_{m=0}^{\infty}$ 
up to absolute error $2^{-m}$ 
to the $s$-th entry of some fixed finite or infinite
complex sequence $(h_{s})_{s\geq0}$.

Let $\runtime_{\calA}(s,m)$ denote the number of steps made by $\calA$ 
in order to produce said error $2^{-m}$ approximation $h_{s,m}$;
for $t\in\IN$ let $\precision_{\calA}(s,t)$ denote the 
(exponent $m$ to base $\tfrac{1}{2}$ of the) absolute error
of $\calA$'s latest output approximation produced within $t$ steps. 
Then it holds \[ \forall s,k,m,t\in\IN: \qquad
\runtime_{\calA}^{1/k}(s,m)\leq t \quad\Leftrightarrow\quad \precision_{\calA}(s,t^k)\geq m \enspace .\]
In particular, for any probability distribution w.r.t. $s$, expected values $\IE_s$ satisfy:
\[ \IE_s\big[ \runtime^{1/k}_{\calA}(s,m) \big] \;\leq\; \calO(m) 
\quad\Longleftrightarrow\quad
\IE_s\big[ \precision_{\calA}(s,t^k) \big] \;\geq\; \Omega(t)  \enspace . \qquad \qed \]
\end{observation}
\subsection{Closure Properties}
\label{ss:Closure}

Naturally all the above notions are closed under pointwise binary sum and by scaling with a polynomial-time computable constant.
The following is well-known about continuous functions:

\begin{fact}
\label{f:ContinuousClosure}
\begin{enumerate}
\item[(1)]
If $t>0$ is polynomial-time computable
and $f\in\calC[0;1]$ is polynomial-time computable with $f(0)=f(1)$,
then its cyclic shift $\tau_t f:x\mapsto f(x+t\bmod 1)$ is again polynomial-time computable.
\item[(2)]
There exists a polynomial-time \Martin{evaluable} $f\in\calC^1[0;1]$
whose derivative $f'$ is not \Martin{evaluable} {\rm\cite[Theorem~6.4]{Ko91}}.
\item[(3)]
If $f\in\calC^1[0;1]$ is polynomial-time \Martin{evaluable}
and its derivative $f'$ satisfies (iii),
then $f'$ is again polynomial-time \Martin{evaluable} {\rm\cite[Theorem~6.2b]{Ko91}}.
\item[(4)]
If $f\in\calC[0;1]$ is \Martin{evaluable} in \classSharpP,
then so is its anti-derivative aka indefinite Riemann integral 
%\begin{equation}\label{e:Integral}
\quad $\int f \;: \;[0;1]\;\ni\; x\;\mapsto\;\int\nolimits_0^x f(y)\,dy\;\in\;\calC^1[0;1]$.
%\enspace . \end{equation}
\item[(5)]
There exists a polynomial-time \Martin{evaluable} $f\in\calC^\infty[0;1]$
such that $\int f$ is not polynomial-time \Martin{evaluable}
unless $\classSharpP\subseteq\classFP$ {\rm\cite[Theorem~5.33]{Ko91}}.
\end{enumerate}
\end{fact}
The Fourier setting is additionally closed under indefinite integration:

\begin{theorem}
\label{t:FourierClosure}
\begin{enumerate}
\item[(1)]
If $t>0$ is polynomial-time computable
and $f\in\ELLtwo[0;1]$ is Fourier-computable in polynomial time,
then its cyclic shift $\tau_t f\in\ELLtwo[0;1]$ is again Fourier-computable in polynomial time.
\item[(2)]
There exists $f\in\ELLtwo[0;1]$ Fourier-computable in polynomial time,
whose weak derivative $f'$ %(in the Sobolev space $\SOB{1}[0;1]$) 
belongs to $\ELLtwo[0;1]$ 
but is not Fourier-computable.
\item[(3)]
If $f\in\ELLtwo[0;1]$ is Fourier-computable in polynomial time
and $f'$ belongs to $\ELLtwo[0;1]$ and satisfies (a\,iii),
then $f'$ is Fourier-computable in polynomial time.
\item[(4)]
If $f\in\ELLtwo[0;1]$ is Fourier-computable in polynomial time,
then its indefinite integral $\int f$
is Fourier-computable in polynomial time as well.
%\item If $f\in\ELLtwo[0;1]$ is Fourier-computable in polynomial time
%and $g\in\calC[0;1]$ is polynomial-time computable,
%then $f\cdot g\in\ELLtwo[0;1]$ is Fourier-computable in polynomial time as well.
%ZZZ wrong!
\end{enumerate}
\end{theorem}
Recall that $\int f$ is H\"{o}lder-continuous for $f\in\ELLtwo[0;1]$;
but the \emph{Fundamental Theorem of Calculus}
``$f=\partial\int f$'' may fail %for $f\in\ELLtwo[0;1]$
on points of discontinuity.
Regarding \cref{t:FourierClosure}(4),
we cannot hope for polynomial-time \Martin{evaluability} of $\int f$
as a continuous function in view of \cref{t:Continuous4}.

\begin{theorem}
\label{t:StepClosure}
\begin{enumerate}
\item[(1)]
If $t>0$ is polynomial-time computable
and $f\in\ELLtwo[0;1]$ is Step-computable in polynomial time,
then its cyclic shift $\tau_t f\in\ELLtwo[0;1]$ is again Step-computable in polynomial time.
\item[(2)]
If $f\in\ELLtwo[0;1]$ is Step-computable in polynomial time,
then its indefinite integral $\int f$
is Step-computable in polynomial time as well.
\end{enumerate}
\end{theorem}

%%%%%%%%%%%%%%%%%%%%%%%%%%%%%%%%%%
\subsection{Comparisons}
\label{ss:Comparing}

This subsection collects both relations and separations among the above notions of computational complexity for \ELLtwo-functions.

\begin{theorem}
\label{t:SharpP}
A function $f\in\ELLtwo[0;1]$ is Fourier-computable 
in \classSharpP (a\myprime) ~iff~ $f$ is Step-computable in \classSharpP on average (c\myprime).
\end{theorem}

\begin{theorem}
\label{t:FourierNotStep}
There exists a function $f\in\ELLtwo[0;1]$ 
which is Fourier-computable in polynomial time (a), 
but $\bar f_{S,0}\in\IC$ is not computable in polynomial time (b)
for $S(n)=2^{n^2+n}$ 
unless $\classSharpP_1\subseteq\classFP_1$.
\end{theorem}

\begin{theorem}
\label{t:StepMeanNotOther}
There exists a function $f\in\ELLtwo[0;1]$ 
which is polynomial time Step-computable 
on average (c), 
but $\hat f_0=\bar f_{1,0}\in\IC$ is not computable in polynomial time (a+b)
unless $\classSharpP_1\subseteq\classFP_1$.
\end{theorem}
%
%Next recall that $\calC[0;1]\subseteq\ELLtwo[0;1]$.

\begin{theorem}
\label{t:Continuous1}
If $f\in\calC[0;1]$ is \Martin{evaluable} in \classSharpP 
according to \cref{d:One}ii\myprime),
then it is also both Fourier-computable in \classSharpP (a\myprime)
and Step-computable in \classSharpP (b\myprime).
\end{theorem}
\cite[Theorem~5.32de]{Ko91} exhibits a
polynomial-time \Martin{evaluable} $f\in\calC^\infty[0;1]$ 
such that $\hat f_0=\bar f_{1,0}$
is not computable in polynomial time (a+b)
unless $\classSharpP_1\subseteq\classFP_1$.

\begin{theorem}
\label{t:Continuous2}
If $f\in\calC[0;1]$ is Step-computable in polynomial time (b)
and has a polynomial modulus of continuity (iii),
then $f$ is \Martin{evaluable} in polynomial time.
\end{theorem}
Recall that a convergent Fourier series need not be continuous,
and the Fourier series of a continuous function
need not converge pointwise.

\begin{theorem}
\label{t:Continuous3}
If $f\in\calC[0;1]$ is Fourier-computable in \classSharpP (a\myprime)
and has a polynomial modulus of continuity (iii),
then $f$ is \Martin{evaluable} in \classSharpP.
\end{theorem}
\begin{corollary}
\label{c:SharpP}
\Cref{t:SharpP,t:Continuous1,t:Continuous3} 
together show: If $f\in\calC[0;1]$ has a polynomial modulus of continuity (\cref{d:One}iii),
then all four conceptions of computability in \classSharpP are equivalent:
(ii\myprime) and (a\myprime) and (b\myprime) and (c\myprime).
\end{corollary}

\begin{theorem}
\label{t:Continuous4}
There exists $f\in\calC^\infty[0;1]$
which is Fourier-computable in polynomial time (a)
but $f(0)$ is not computable in polynomial time
unless $\classSharpP_1\subseteq\classFP_1$.
\end{theorem}

%%%%%%%%%%%%%%%%%%%%%%%%%%%%%%%%%%%%
\subsection{Proofs of Comparisons and Closure Properties}
\label{ss:Proofs}

\renewcommand*{\thesubsubsection}{\thesubsection.\arabic{subsubsection}}

%%%%%%%%%%%%%%%%%%%%%%%%%%%%%%%%%%%%
\subsubsection{Proof of \Cref{t:SharpP}: \classSharpP-Equivalence in \ELLtwo}
\label{sss:SharpP}

Recall that (a\,iii) and (b\,iii) are equivalent by \cite[Theorem~17c+a]{Aras25};
and that $\classGapP$ is closed under polynomial products and under exponential sums
\cite[Lemma~3.1]{CountingComplexity}. %ZZZ
Also abbreviate 
\begin{equation}
\label{e:Twiddle}
e_{S,s,k} \;:=\; \big(\exp(2\pi ik \tfrac{s+1}{S}) - \exp(2\pi i k \tfrac{s}{S})\big)\cdot S/(2\pi i k) \;\in\IC 
\enspace .
\end{equation}
These complex coefficients are polynomial-time computable \cite[\S7.3]{Wei00}
\cite[Example~4]{Aras25}; more precisely:
There exist Gaussian integers $e_{S,s,k,n}\in\IZ+i\IZ$,
computable in time polynomial in $n+\log k+\log S$,
such that $|e_{S,s,k}-e_{S,s,k,n}/2^n|\leq2^{-n}$; 
for convenience let $e_{S,s,0}:=1$ and $e_{S,s,0,n}:=2^n$,
so that $|e_{S,s,k}|\leq1$ for all $s,k$. 

\medskip
First suppose (a) that $f\in\ELLtwo[0;1]$ is Fourier-computable in \classSharpP:
$| \hat f_k - c_{k,n}/2^n | \leq 2^{-n}$ for both 
$\TWO^*\ni\sdone^n\,\sdzero\,\bin(k)\mapsto \Re c_{k,n},\Im c_{k,n}\in\IZ$ 
in $\classGapP$.
%\[  | \hat f_k - c_{k,n}/2^n | \;\leq\; 2^{-n}, \quad
%\Re \big(c_{k,n}+C\big)\;=\; \Card\big\{\vec v\in\TWO^*: |\vec v|\leq\poly(n), \vec v\in P\big\}, \quad
%\Im \big(c_{k,n}+C\big)\;=\; \Card\big\{\vec w\in\TWO^*: |\vec w|\leq\poly(n), \vec w\in q\big\}, 
%\quad P,Q\in\classP \enspace . \]
%Following the proof of \cite[Theorem~17c]{Aras25}, let
Record that
$\bar f_{S,s}= \hat f_0+\sum\nolimits_{k\neq0} \hat f_k \cdot e_{S,s,k}$.
Now let $h_{S,s,K,n}:=\sum_{|k|\leq K} c_{k,m}\cdot e_{S,s,k,m'}$
for $m:=n+\log K+\calO(1)$ and $m':=n+\log K+\log F+\calO(1)$,
where $F\in\IN$ is some constant (depending only on $f$)
such that $\max_k|\hat f_k|\leq F$.
Then closure of $\classGapP$ under binary multiplication and under exponential sums
implies that $h_{S,s,K,n}$ belongs to \classGapP,
since $\log S,\log K\leq\poly(n)$ by (a\,iii) and by \cite[Theorem~17ac]{Aras25}:
establishing (b'\,ii).
\cite[Theorem~17ac]{Aras25} also yields
\[ 
4^{-n} 
\;\geq\; \big\|\calS_S\big(\calF_K(f)\big)-\calF_K(f)\big\|_2^2 
\;=\;\sum\nolimits_{s=0}^{S-1} \big|\bar f_{S,s}-h_{S,s,K}\big|^2/S
\]
for $h_{S,s,K}:=\sum\nolimits_{|k|\leq K} \hat f_k \cdot e_{S,s,k}$.
Moreover 
\begin{eqnarray*}
|h_{S,s,K}-h_{S,s,K,n}/2^{m+m'}|
&\leq& \sum_{|k|\leq K} |c_k-c_{k,m}/2^m|\cdot |e_{S,s,k}| \\
& +& \sum_{|k|\leq K} |c_k|\cdot |e_{S,s,k}-e_{S,s,k,m'}/2^{m'}| \\
&\leq& (2K+1)\cdot 2^{-m+1} \;+\; (2K+1)\cdot F \cdot 2^{-m'} 
\end{eqnarray*}
$\leq 2^{-n}$ by the above choices of $K,S$ and $m,m'$:
establishing (c\,i), and thus
Step-computability of $f$ in \classSharpP in mean (c').

\medskip
Now suppose conversely (c') that $f\in\ELLtwo[0;1]$ is Step-computable in \classSharpP in mean
with $h_{S,s,n}$ in \classGapP for $S=\poly(n)$ according to (b\,iii).
Then $c_{k,S}:=-\sum_{s=0}^{S-1} \bar f_{S,s}\cdot e_{S,s,k}$ has
\[ |\hat f_k-c_{k,S}|^2 \;\leq\; \sum\nolimits_k |\hat f_k-c_{k,S}|^2 \;=\;
\|f-\calS_S(f)\|_2^2  \;\leq\; 4^{-n} \]
by Parseval and (b\,iii).
Moreover $c_{k,S,n}:=-\sum_{s=0}^{S-1} h_{S,s,m}\cdot e_{S,s,k,m'}$ belongs to \classGapP.
And 
\begin{eqnarray*}
|c_{k,S}-c_{k,S,n}/2^{m+m'}|
&\leq& \sum\nolimits_{s=0}^{S-1} |\bar f_{S,s}-h_{S,s,m}/2^m| \cdot |e_{S,s,k}| \\
& +& \sum\nolimits_{s=0}^{S-1} |\bar f_{S,s}|\cdot |e_{S,s,k}-e_{S,s,k,m'}/2^{m'}| \\
&\leq& S\cdot 2^{-m} \;+\; S\cdot 2^{-m'}
\end{eqnarray*}
$\leq 2^{-n}$ for the above choices of $K,S$ and $m,m'$:
establishing (a\,i), and thus Fourier-computability of $f$ in \classSharpP (a).

%%%%%%%%%%%%%%%%%%%%%%%%%%%%%%%%%%%%
\subsubsection{Encoding Discrete Problems into the Reals}
\label{sss:Encode}

In \cref{sss:FourierNotStep,sss:StepMeanNotOther}, we shall encode discrete decision
and counting problems into real numbers and functions---in
such a way that efficient approximation of the latter
allows to recover the former. 
Similar techniques are well-known \cite{Spe49,PR79,Ko82,Kaw10}.
For the unfamiliar reader we recall some basic subtleties of such encodings.
Recall that computing a real number means approximating it by dyadic rationals 
up to any given error bound $2^{-n}$.

\begin{remark}
\label{r:Encode}
Fix a discrete decision problem $P\subseteq\IN$,
and let $r_P:=\sum_{N\in P} 2^{-N}\in[0;2]$ denote the
real number whose binary expansion is the characteristic function of $P$.
\begin{enumerate}
\item[a)]
 $P$ is decidable ~iff~ $r_P$ is computable
 \cite[Theorem~2.3]{Ko91}, \cite[Lemma~4.2.1]{Wei00}.
\item[b)]
 The equivalence (a) is not uniform:
 No algorithm can decide $P$ from given approximations to $r_P$
\cite{Tur37}, \cite[Theorem~4.1.13.3]{Wei00}.
\item[c)]
 The equivalence (a) does not extend to computational complexity:
 There exist polynomial-time computable $r_P$ 
 such that $P\not\in\classP$, 
 regardless of the Millennium Prize Problem \cite[Theorem~2.8]{Ko91}.
\item[d)]
Proceeding from binary to ternary encoding,
let $r'_P:=\sum_{N\in P} 3^{-N}$. 
Then $r'_P$ is polynomial-time computable ~iff~ $P\in\classP$;
and this equivalence holds uniformly.
\end{enumerate}
To see (d), note that 
\[ 0\;\in\; P \quad\Longleftrightarrow\quad r'_P\;\geq\;1, \qquad
0\;\not\in\; P \quad\Longleftrightarrow\quad r'_P\;\leq\;1/2 \enspace : \]
the two cases can be distinguished by approximating $r'_P$ up to error $<\tfrac{1}{2}$.
Forcing $0\not\in P$ by replacing $r'_P$ with $r'_P-1$ in case $0\in P$, 
``$1\in P$?'' can be recovered next. Indeed it holds 
\begin{equation}
\label{e:WellSep0}
N\in P \;\Leftrightarrow\; r'_P\geq3^{-N}, \quad
N\not\in P \;\Leftrightarrow\; r'_P\leq3^{-N}/2=\sum\nolimits_{n>N} 3^{-n}
\end{equation}
provided that $0,1,\ldots,N-1\not\in P$.
The latter (at most $N$-fold) repeated conditional subtractions of $3^{-0},3^{-1},\ldots,3^{-N+1}$ from $r'_P$
are performed with precision $4^{-N}$ each: such that the 
two cases in \cref{e:WellSep0} remain numerically
`well-separated' in spite of error propagation.
\qed\end{remark}
\Cref{e:WellSep0} says that all $n$-th terms, $n>N$,
can sum up to at most half the $N$-th term; which in turn
affects only the first $\poly(N)$ bits \ie up to error $2^{-\poly(N)}$.
In the sequel we ensure similar numerical `well-separations' (*) when devising
more sophisticated encodings of discrete counting problems into \ELLtwo-functions.
%
%%%%%%%%%%%%%%%%%%%%%%%%%%%%%%%%%%%%
\subsubsection{Proof of \Cref{t:FourierNotStep}: Fourier but not Step Polytime}
\label{sss:FourierNotStep}

Recall that
\begin{equation}\label{e:Twiddle2}
\bar f_{S,0} \;=\; \hat f_0+\sum\nolimits_{k\neq0} \hat f_k \cdot e_{S,0,k},  \quad
e_{S,0,k} \;=\; \big(\exp(2\pi ik /S) - 1\big)\cdot \tfrac{S}{2\pi i k}  \end{equation}
according to \cref{e:Twiddle}.
We also record $|e_{S,0,k}|\leq1$ and,
for $|k/S|\ll 1$, $e_{S,0,k}\approx 1$ by L'H\^opital; 
more precisely
\begin{equation} 
\label{e:ExpEstim}
|k/S|\leq1 \quad\Longrightarrow\quad |e_{S,0,k}-1|\leq |k/S| \enspace . 
\end{equation}
We also record
\[ \sum\nolimits_{n\geq N} 2^{-n^2}
\;=\; \sum\nolimits_{m\geq 0} 2^{-(N^2+2mN+m^2)}
\;\leq\; 2^{-N^2+\calO(1)} \enspace . \]
Now to construct $f$ as claimed in \cref{t:FourierNotStep},
fix some $\classSharpP_1$-complete $\varphi:\IN_+\to\IN$ w.l.o.g. of the form
\begin{equation}
\label{e:SharpPone}
\varphi(m) \;=\; \Card\big\{ \vec w\in\TWO^m: \vec w\in P\big\}, \quad P\in\classP 
\end{equation}
and let $f$ be the Fourier series with coefficients 
\begin{equation}
\label{e:FourierCount}
\hat f_k \;:=\; 2^{-m^2} \cdot \left\{ \begin{array}{r@{\;:\;}l}
1 & \bin_m(w)\in P \\
0 & \bin_m(w)\not\in P 
\end{array} \right\} , \quad 
k=2^m+w, \; w\in\{0,1,\ldots,2^m-1\} %\enspace ,
\end{equation}
where 
\[ \bin_m\Big(\sum\nolimits_{j=0}^{m-1} b_j 2^j\Big) \;:=\; (b_0,\ldots,b_{m-1})\;\in\;\TWO^m  \]
and $\hat f_k:\equiv0$ for $k\leq1$.
Said coefficient sequence $\hat f_k$ 
satisfies \cref{d:Two}a\,iii),
and in fact is polynomial-time computable (\cref{d:Two}a\,i+ii)
since $P\in\classP$ and $\log|k|=\calO(m)$;
moreover it is square summable,
thus giving rise to an $\ELLtwo$-function $f$.

Intuitively, in view of Equations~\eqref{e:Twiddle2} and \eqref{e:ExpEstim} and \eqref{e:FourierCount},
$\bar f_{S,0}\approx\sum\hat f_k$ should allow 
for recovering $\varphi(N)$ %from \cref{e:SharpPone} 
from $\bar f_{S,0}$ 
within time polynomial in $N$ for all sufficiently large $N$,
and to thus conclude the proof. 
To make this idea rigorous note that, in $\sum_k\hat f_k$,
the case $\varphi(N)\geq1$ 
is `well-separated' (\cref{sss:Encode})
from the case $\varphi(N)=0$;
namely the weight $\hat f_{2^N+W}$ of one witness
$W$ to $N\in P$ is larger than the weight of all 
possible witnesses $w$ to all $n>N$ combined:
For $K:=2^N+W$ with $0\leq W<2^N$,
\begin{multline*}
\bin_N(W)\in P 
\quad\Rightarrow\quad 
\hat f_{K}=2^{-N^2}
\overset{(*)}{\;\gg\;}
2^{-N^2-N+\calO(1)}  \\
= \; \sum\limits_{m>N} 2^m\cdot 2^{-m^2}
\;\geq\;
\sum\limits_{m>N} \sum\limits_{w=0}^{2^m-1} \hat f_{2^m+w} 
\;=\;\sum\limits_{k\geq2^{N+1}} \hat f_k
\enspace . \end{multline*} 
Moreover the two cases are also well-separated
in $\bar f_{S,0}$ for $S:=2^{N^2+N}$: 
%Let $K:=2^{N+1}$ to include all candidate witnesses to $\varphi(N)$ and estimate
\begin{multline*} 
\Big| \bar f_{S,0} - \sum\limits_{k<2^{N+1}} \hat f_k \Big|
\;\;\leq\;\;
\sum\limits_{k<2^{N+1}} \underbrace{\hat f_k}_{\leq 1}\cdot \underbrace{\big|e_{S,0,k}-1\big|}_{\leq k/S \text{ by \eqref{e:ExpEstim}}}
\;+\;
\sum\limits_{k\geq 2^{N+1}} |\hat f_k| \cdot \underbrace{|e_{S,0,k}|}_{\leq1} \\
\leq\;\;
2^{N+1}\cdot (2^{N+1}/2^{N^2+N}) \;+\; \sum\nolimits_{k\geq2^{N+1}} |\hat f_k| 
\quad\overset{(*)}{\leq\;} 2^{-N^2-N+\calO(1)} \enspace .
\end{multline*}

%%%%%%%%%%%%%%%%%%%%%%%%%%%%%%%%%%%%
\subsubsection{Proof of \Cref{t:StepMeanNotOther}: Step in Mean but not otherwise Polytime}
\label{sss:StepMeanNotOther}

\begin{figure}[htb]
\includegraphics[width=0.98\textwidth]{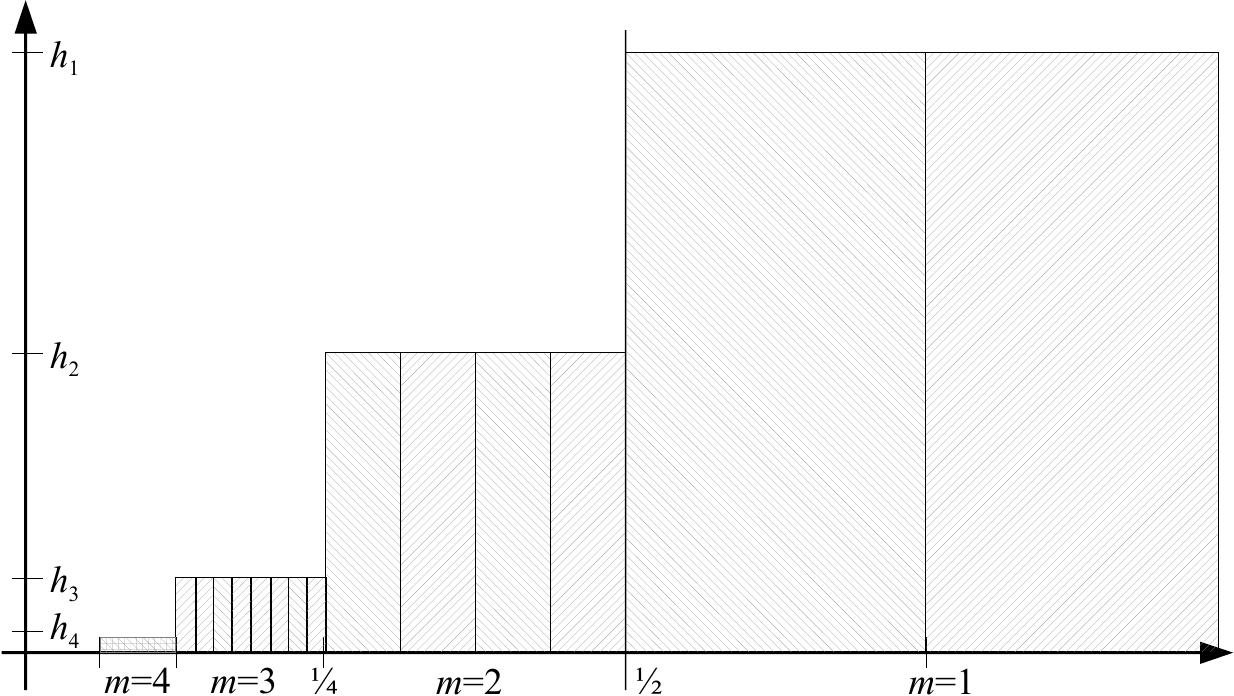}\vspace*{-2ex}%
\caption{\label{f:StepFunc}Illustrating the construction of $g$ from the proof of \cref{t:StepMeanNotOther}}
\end{figure}

We adopt the construction from \cite[\S3]{Fri84}; 
see also the proof of \cite[Theorem~5.32]{Ko91}.
To this end fix some $\classSharpP_1$-complete $\varphi:\IN\to\IN$ w.l.o.g. of the Form~\eqref{e:SharpPone}.
Encode $\varphi$ into the real number $\int_0^1 g(t)\,dt$
for a piece-wise constant function $g:[0;1]\to[0;1]$ 
with infinitely descending step widths as follows:
Divide $(0;1]$ into adjacent sub-intervals $I_m:=(2^{-m};2^{-m+1}]$, 
one for each possible argument $m\in\IN_+$ to $\varphi$;
and sub-divide each $I_m$ further into $2^m$ sub-sub-intervals $I_{m,\vec w}$
of width $4^{-m}$, one for each possible $\vec w\in\TWO^m$; see \cref{f:StepFunc}.
Make the restriction $g_{m,\vec w}:=g|_{I_{m,\vec w}}$ identically zero in case $\vec w\not\in P$,
while $g_{m,\vec w}:\equiv2^{-m^2}$ in case $\vec w\in P$.
The areas thus assigned to each argument $m$ and witnesses $\vec w$ 
satisfy the following well-separation:
\[
4^{-N}\cdot 2^{-N^2} \overset{(*)}{\;\gg\;}
4^{-N^2-3N+\calO(1)} \;=\;
\sum\nolimits_{n>N} \sum\nolimits_{w=0}^{2^n-1} 4^{-n}\cdot 2^{-n^2} \enspace .
\]
This shows that $\varphi(N)$ can be recovered from
the total area $\int_0^1 g(t)\,dt=\hat f_0\in\IC$
in time polynomial in $N$.
We now establish Step-computability of $g$ in polynomial time
(\cref{d:Two}b), thus concluding the proof of
\cref{t:StepMeanNotOther}.
Note that \cref{d:Two}b) requires approximation by
steps of given width $1/S$, so we need to efficiently 
average over the `thin' steps comprising $g$,
although said area encodes $\varphi$.
For Step $s=0$, this is $\classSharpP_1$-hard;
recall (the proof of) \cref{t:FourierNotStep}.
But rather than (b) worst-case approximation for each $s$,
we only need to show approximation in mean (c).
To this end fix $N\in\IN$ and let $S:=2^N$:
For all $m\leq N/2$ and all $w<2^m$, 
$I_{m,w}$ is equal to or contains some step 
$\big(s/S;(s+1)/S\big]$ of $\calS_S(f)$,
whose height $\bar h_{S,s}=g_{m,\vec w}$ 
thus can simply be read off exactly 
within time polynomial in $m\leq\calO(N)$.
For all $m\geq\sqrt{N}$ on the other hand, 
$|g_{m,\vec w}|\leq2^{-m^2}$
is trivially approximated in mean
by steps $\bar h_{S,s}$ of height zero.

%%%%%%%%%%%%%%%%%%
\subsubsection{Proof of \Cref{t:Continuous1}: \classSharpP-Equivalence to the Continuous Case}
\label{sss:Continuous1}

Let us write %ZZZ
\[
\calL_{L,n}(\bar v,t) \;:=\; 
%v_{L,\ell,n}/2^n \;+\; (L\cdot t-\ell)\cdot \big(v_{n,L,\ell+1}-v_{n,L,\ell}\big)/2^n \;=\; 
(L\cdot t-\ell)\cdot v_{n,L,\ell+1}/2^n \;+\; (1-L\cdot t+\ell)\cdot v_{n,L,\ell}/2^n, \quad
\tfrac{\ell}{L}\leq t\leq \tfrac{\ell+1}{L}  \]
for the piecewise linear function induced by $\bar v$, and record
\begin{eqnarray}
\label{e:LinearAverage}
\int\nolimits_{\sigma}^{\tau} \calL_{L,n}(\bar v,t)\,dt
&=& (\tau-\sigma)\cdot\big(L\cdot \tfrac{\tau+\sigma}{2}-\ell\big)\cdot v_{n,L,\ell+1}/2^n  \\ \nonumber
&+& (\tau-\sigma)\cdot\big(1-L\cdot\tfrac{\tau+\sigma}{2}+\ell\big)\cdot v_{n,L,\ell}/2^n
, \quad \tfrac{\ell/L}\leq \sigma\leq \tau \leq \tfrac{\ell+1}{L} 
\end{eqnarray}
\begin{eqnarray} \label{e:LinearFourier} %\nonumber
\widehat{\calL_{L,n}(\bar v)}_k 
&=& \sum\limits_{\ell=0}^{L-1} 
\big( (\ell+1) \cdot v_{n,L,\ell} \:-\: \ell\cdot v_{n,L,\ell+1} \big) \\ \nonumber %\label{e:LinearFourier}
&& \qquad\quad \cdot\;
\big( \exp(-2\pi i k \tfrac{\ell+1}{L}) - \exp(-2\pi i k \tfrac{\ell}{L}) \big)/(2\pi i k\cdot 2^n) \\ \nonumber
&+& \sum\limits_{\ell=0}^{L-1}
\big( \exp(-2\pi i k \tfrac{\ell+1}{L})\cdot (1+2\pi i k \tfrac{\ell+1}{L})
- \exp(-2\pi i k \tfrac{\ell}{L})\cdot (1+2\pi i k \tfrac{\ell}{L})\big) 
\\ \nonumber
&& \qquad\quad \cdot\;
L\cdot \frac{v_{n,L,\ell+1}-v_{n,L,\ell}}{4\pi^2 k^2\cdot 2^n},
\qquad k\neq 0 \\[0.5ex] \nonumber
\widehat{\calL_{L,n}(\bar v)}_0 
&=& \sum\nolimits_{\ell=0}^{L-1} 
\big( (\ell+1) \cdot v_{n,L,\ell+1} \:-\: \ell\cdot v_{n,L,\ell} \big)/(L\cdot 2^{n}) 
\enspace .
\end{eqnarray}
Similarly to the proof of \cref{t:SharpP},
$\widehat{\calL_{L,n}(\bar v)}_k$ 
and 
$\widebar{\calL_{L,n}(\bar v)}_{S,s}$ 
can thus be seen computable in \classSharpP
whenever $\bar v$ is,
due to the aforementioned closure properties of \classGapP.
And (iii) implies 
$\big|\hat f_k-\widehat{\calL_{L,n}(\bar v)}_k\big|\leq 2^{-m}$
for $m=n+\log k+\log L+\calO(1)$ and $L=2^{\poly(m)}$; similarly
$\big|\bar f_{S,s}-\widebar{\calL_{L,n}(\bar v)}_{S,s}\big|\leq 2^{-m}$
for $m=n+\log S+\log L+\calO(1)$ and $L=2^{\poly(m)}$.

%iii) It holds $\|f-\calL_L(f)\|_\infty \leq 2^{-m}$ for $L=2^{\poly(m)}$, where
%\[ \calL_L(f,t) \;:=\; 
%f(\ell/L) \;+\; \big(L\cdot t-\ell\big)\cdot \big(f(\tfrac{\ell+1}{L})-f(\tfrac{\ell}{L})\big),
%\quad \tfrac{\ell}{L}\leq t\leq \tfrac{\ell+1}{L} \]
%denotes the piecewise linear function interpolating $f$ with $L$ equidistant breaking points. 

%%%%%%%%%%%%%%%%%%
\subsubsection{Proof Sketches of \Cref{t:Continuous2,t:Continuous3,t:Continuous4}}
\label{sss:Continuous2}

Having spelled out in full detail the proofs of the above theorems,
we can now take a step back and focus on proof ideas
for the remainder.

\begin{proof}[\cref{t:Continuous2}]
Record that
\begin{eqnarray*}
\big| f(\ell/L) - \bar f_{L,\ell} \big|
&=&
\big| L\cdot\int\nolimits_{\ell/L}^{(\ell+1)/L} f(\ell/L) - f(t)\, dt\big| \\
&\leq&
L\cdot\int\nolimits_{\ell/L}^{(\ell+1)/L} |f(\ell/L) - f(t)| \, dt 
\quad\leq\; 2^{-n}
\end{eqnarray*}
for $L:=2^{\poly(n)}$:
(i)+(ii) thus follow from (b\,i) and (b\,ii).
\qed\end{proof}
Next record from \cite[\S I.3 Corollary~II]{Jackson30}:

\begin{fact}
\label{f:JacksonII}
If $f\in\ELLtwo[0;1]$ has a polynomial modulus of continuity (iii), then
\[ \big\| f - \calF_{K}(f)\big\|_\infty \;\leq\; \calO(n\cdot 2^{-n}) 
\quad\text{ for } K=2^{\poly(n)} \enspace . \]
\end{fact}
\begin{proof}[\cref{t:Continuous3}]
Summing the first $2^{\poly(n)}$ terms of the Fourier series
thus approximates $f$; and incurs complexity in \classSharpP
due to the aforementioned closure properties of \classGapP.
\qed\end{proof}

\begin{proof}[\cref{t:Continuous4}]
Recall that our proof of \cref{t:FourierNotStep} encodes
a $\classSharpP_1$-hard problem into $\sum_k\hat f_k=f(0)$,
where $0\leq \hat f_k\leq k^{-\log k}$
according to \cref{e:FourierCount}.
In particular $|(2\pi ik)^d\cdot\hat f_k|\leq k^{d-\log k}$ 
is absolutely summable for every fixed $d$:
Recall that $2\pi i k\hat f_k$ are the Fourier Coefficients 
of the first derivative $f'$.
Therefore the function $f\in\ELLtwo[0;1]$ 
constructed in our proof of \cref{t:FourierNotStep}
is $d$-times differentiable for every $d\in\IN$.
\qed\end{proof}

%%%%%%%%%%%%%%%%%%%%%%%%%%%%%%%%%%%%
\subsubsection{Proof of \Cref{t:FourierClosure}: Properties of Fourier Computability}

\begin{enumerate}
\item
Abbreviate $g:=\tau_t f$.
Since $\tau_t$ is an isometry, the phase-shifted Fourier 
coefficients $\hat g_k:=\hat f_k\cdot\exp(-2\pi i k t)$
again satisfy (a\,iii).
And $\hat g_k$ 
can be approximated as well (a\,i) and as efficiently (a\,ii) as the original ones $\hat f_k$.
\item
Similarly to a \emph{Specker Sequence},
let $\hat g_k:=2^{-M_k}\leq1$ where $M_k\in\IN$ denotes the $k$-th terminating Turing machine
($k=1,2,3,\ldots$)
in some recursively enumerated order without repetition.
Then $\hat g_k$ can be approximated up to error $2^{-n}$
in time polynomial in $n+\log k$: (a\,i+ii).
Moreover $\sum_k |\hat g_k|^2$ converges, 
but has no polynomial nor even recursive rate of convergence (a\,iii).
According to Parseval, the Fourier series $g$ with coefficients $\hat g_k$ 
belongs to $\ELLtwo[0;1]$ but is not Fourier computable.
Its anti-derivative $f:=\int g\in\calC[0;1]$ has
Fourier coefficients $\hat f_k=\hat g_k/k$ for $k\neq0$.
Again $\hat f_k$ can be approximated in time polynomial in $n+\log k$.
Moreover $\sum_k|\hat f_k|^2$ converges at least as fast as $\sum_k 1/k^{2}\leq\pi^2/6$
and thus satisfies (a\,iii).
\item
As before, if $\hat f_k$ can be approximated well (a\,i) and efficiently (a\,ii),
then so can $\hat g_k:=2\pi ik\cdot \hat f_k$:
the Fourier coefficients to $g=f'$.
\item
Recall that $\int f$ is H\"{o}lder-continuous 
(iii)/(a\,iii) for $f\in\in\ELLtwo[0;1]$.
Also, the Fourier series of the identity $\id:[0;1]\to[0;1]$ is 
\begin{equation}
\label{e:FourierID}
\id(x)\;\;=\;\; \sum\limits_k \widehat\id_k\cdot\exp(2\pi i k x), \qquad
\widehat\id_0=\tfrac{1}{2}, \quad \widehat\id_k =-\frac{1}{2\pi i k} \; (k\neq0)
\end{equation}
and Fourier-computable in polynomial time (a).
$\int f$ has Fourier coefficients
\[ \widehat{\big(\int f\big)}_0 \;=\; \hat f_0/2, \qquad
\widehat{\big(\int f\big)}_k \;=\; \hat f_k/(2\pi i k) - \hat f_0/(2\pi i k), \; k\neq 0
\]
computable in polynomial time (a\,i)+(a\,ii);
and summable at least as fast as $\hat f_k$ (a\,iii).
\qed
\end{enumerate}

%%%%%%%%%%%%%%%%%%%%%%%%%%%%%%%%%%%%
\subsubsection{Proof of \Cref{t:StepClosure}: Properties of Step Computability}

\begin{enumerate}
\item follows from Item~2 since 
\[ \widebar{\tau_r f_{S,s}} \;=\; S\cdot\int\limits_{0}^{r+(s+1)/S\bmod 1} f(t)\,dt
\;-\; S\cdot\int\limits_{0}^{r+(s/S)\bmod 1} f(t)\,dt \]
yields (b\,i) and (b\,ii); and (b\,iii) is invariant under cyclic shift anyway.
\item
Let $t=\sum_{j\geq1} b_j2^{-j}\in[0;1]$, $b_j\in\TWO$. 
The disjoint decomposition $\big[0;t\big)=$
\begin{multline}  \big[0;b_1/2\big) \;\cup\; \big[2b_1/4;(2b_1+b_2)/4\big) \;\cup\; 
\big[(4b_1+2b_2)/8;(4b_1+2b_2+b_3)/8\big) \cup\cdots 
\\ \label{e:DyadicDecomposition}
\cup\; \big[\sum\nolimits_{j=1}^J b_j/2^j;\sum\nolimits_{j=1}^{J+1} b_j/2^j\big) \;\cup \ldots 
\end{multline}
gives rise to $\int\limits_0^t f(x)\,dx=$
\begin{equation}
\label{e:DisjointSum}
b_1\cdot \widebar f_{2,0} \;+\; b_2\cdot \widebar f_{4,2b_1} \;+\;
b_3\cdot \widebar f_{8,(4b_1+2b_2)} \;+\cdots+\; b_j\cdot\widebar f_{2^j,(2^{j-1}b_1+\cdots 2b_{j-1})} \;+\ldots 
\end{equation}
Hence replacing real $t>0$ in \cref{e:DyadicDecomposition}
with dyadic approximation $t_{m}/2^{m}$ ($t_m\in\IN$)
makes Series~\ref{e:DisjointSum} 
boil down to $m$ (=only polynomially many)
terms of the form $\widebar f_{2^j,s_j}$;
each of which can be calculated (and then added) 
in polynomial time by hypothesis (b); and
whose sum approximates $\int_0^t f(x)\,dx$
up to error $\leq F\cdot 2^{-m/2}$
for any $F\geq\|f\|_2$:
Recall that $\int|f|$ is H\"{o}lder-continuous,
$\int_x^y |f(t)|\,dt\leq F\cdot\sqrt{|y-x|}$.
\qed\end{enumerate}

%%%%%%%%%%%%%%%%%%%%%%%%%%%%%%%%%%%%%%%%%%%%%%%%%
\section{Applications}
\label{s:Applications}
In this section we investigate applications of Definition \ref{d:Two} to two important PDE examples: Heat (parabolic) and Wave (hyperbolic) equations.
%%%%%%%%%%%%%%%
\subsection{Heat Equation}\label{ss:Heat}

Consider an initial-value problem for the heat equation on the unit interval with periodic boundary conditions:
\begin{equation} \label{eq:heat}
\begin{cases}
u_t = u_{xx}, \ x \in [0,1) \text{ mod } 1, \ t > 0, \\   
u(x,0) = f(x),
\end{cases}
\end{equation}
where $f\in\ELLtwo[0,1]$. As is well known, this problem has a unique solution $u$, and for any fixed $t>0$, $u(\cdot,t)\in C^{\infty}[0,1]$ (hence also in $\SOB{s}[0,1]$ for all $s>0$). Recall that  $r\in\mathbb R$ is \emph{polynomial-time computable} if there is an algorithm that, given $n \in \mathbb{N}$, outputs, within time polynomial in $n$, a sequence of  $r_n\in\mathbb Z$ such that $|r - r_n/2^n| \leq 2^{-n}$.

\begin{theorem}\label{t:heat}
    Let $t > 0$ be a polynomial-time computable real number. \newline
     If the initial function $f \in\ELLtwo[0,1]$ is Fourier-computable in polynomial time, then the solution $u(\cdot, t)$ at time $t$ is also Fourier-computable in polynomial time. %\Sveta{[Also, point-wise polynomial-time computable in the classical sense, it looks like.]}.
\end{theorem}
\begin{proof}
    The solution of \eqref{eq:heat} is given by
%\begin{equation} \label{eq:solution}
$u(x,t) = \sum_{k \in \mathbb{Z}} \hat{f}_k e^{-4\pi^2 k^2 t} e^{2\pi i k x}.$
%\end{equation}
Thus, the Fourier coefficients of $u(\cdot, t)$ are $\hat{u}_k(t) = \hat{f}_k e^{-4\pi^2 k^2 t}$.

Since $f$ is Fourier-computable in polynomial time, there exist a polynomial $p$ and a double sequence $(c_{k,n})$ of Gaussian integers, such that a\,i), a\,ii) of Definition \ref{d:Two} hold, and  $\big\|f - \mathcal{F}_K(f)\big\|_2 \leq 2^{-n}$ with $K = 2^{p(n)}$.

By Parseval's identity,
$$\big\|u(\cdot,t) - \mathcal{F}_K(u(\cdot,t))\big\|_2^2
= \sum_{|k|>K} |\hat{u}_k(t)|^2
= \sum_{|k|>K} |\hat{f}_k|^2 e^{-8\pi^2 k^2 t}
\leq \sum_{|k|>K} |\hat{f}_k|^2 \leq 2^{-2n},$$
hence $\big\|u(\cdot,t) - \mathcal{F}_K(u(\cdot,t))\big\|_2 \leq 2^{-n},$ 
therefore (a\,iii) for $u(\cdot,t)$ holds for the same $K = 2^{p(n)}$.

Now we need to construct a double sequence $(d_{k,n})$ of Gaussian integers, such that (a\,i) and (a\,ii) hold for $u(\cdot,t)$. 
Since $t$ is polynomial-time computable, we can compute, within polynomial time, an integer sequence $t_m$ such that $|t-t_m/2^m|\leq 2^{-m}$. It is also possible to compute, within polynomial time, an integer sequence $b_{k,m}$ (i.e., a sequence of Gaussian integers with zero imaginary part) such that
$\bigl| e^{-4\pi^2  k^2 t} - b_{k,m}/2^m \bigr| \leq 2^{-m}$,
since the exponential function is polynomial-time computable, as well as the argument $4\pi^2  k^2 t$. We also have $\big|\hat f_k-c_{k,m}/2^m\big|\leq2^{-m}$.   

From  a\,iii) with $n=0$  we have $K_0 = 2^{p(0)}$ such that $\big\|f - \mathcal{F}_{K_0}(f)\big\|_2 \leq 1$, hence 
$\big\|f\big\|_2 \leq \big|\mathcal{F}_{K_0}(f)\big|_2 + 1$. Since $K_0$ is a constant, we can compute a rational number $M$ with $\big\|f\big\|_2\leq M$ in constant time (by approximating the finitely many coefficients $\hat{f}_k$ for $|k|\leq K_0$ using the given $c_{k,1}$).  

Define
$\alpha_n = \Bigl\lceil \sqrt{ \frac{\ln M + n\ln 2}{4\pi^2  t} } \Bigr\rceil .
$
For $|k| \geq \alpha_n$ we have $|\hat{u}_k(t)| \leq M e^{-4\pi^2  k^2 t} \leq 2^{-n}$; set $d_{k,n}=0$ for these $k$.

For $|k| < \alpha_n$ choose an auxiliary precision parameter
$m = n + \lceil \log_2(M+2) \rceil + 1$.

For each such $k$, let $P_{k,m} = c_{k,m} \cdot b_{k,m} \in \mathbb{Z}[i]$ and define $d_{k,n}$ by rounding $P_{k,m}/2^{2m-n}$ to the nearest Gaussian integer:
    $d_{k,n} = \Bigl\lfloor \frac{P_{k,m}}{2^{2m-n}} \Bigr\rceil.$
    Then
    $
    \bigl| \hat{u}_k(t) - d_{k,n}/2^n \bigr| \leq 2^{-n}.
    $
    
Finally, let us check Condition~(a\,ii) for $u$.  
For $|k| \geq \alpha_n$, $d_{k,n}=0$. For $|k| < \alpha_n$, note that $\alpha_n = O(\sqrt{n})$, hence $\log|k| = O(\log n)$.  
Computing $c_{k,m}$ takes time polynomial in $m+\log|k| = \text{poly}(n)$.  
Computing $b_{k,m}$ takes time polynomial in $m$. The multiplication and rounding are polynomial in the bit-length of $P_{k,m}$, which is $O(m)$.  
Thus $(n,k) \mapsto d_{k,n}$ is computable in time polynomial in $n+\log|k|$.\qed
\end{proof}

\begin{remark}
\label{r:Blowup}
    In contrast to \Cref{t:heat}, \cite[Theorem~4e/\S4.3]{KPSZ23} has constructed 
a polynomial-time computable initial condition $f\in \calC^\infty$ 
such that $u(\cdot,1)$ is continuous but not polynomial-time computable 
as a continuous function unless $\classSharpP_1\subseteq\classFP_1$;
cmp. also \cite{bbk23complexity}.
\end{remark}

\begin{proposition}    \label{t:heat_step}
There exists $f\in L^2[0,1]$ which is Step-computable in polynomial time on average (Definition \ref{d:Two}c), %with %$S_0(m)=2^m$, 
such that for every polynomial-time computable $t>0$, the solution $u(\cdot,t)$ is not Step-computable in polynomial time on average with $S_0(m)=2^{p(m)}$ and any $p(m)=am+b$, $0<a<2$, unless $\classSharpP_1\subseteq \classFP_1$.
\end{proposition}

\begin{lemma}\label{l:analyt-stepmean}
Let $f$ be Lipschitz continuous. Suppose $f$ is Step-computable in polynomial time on average (Definition \ref{d:Two}c) with $S_0(m)=2^{p(m)}$, where $p(m)=am+b$ with $0<a<2$. Then $f$ is classically polynomial-time computable (Definition \ref{d:One}).
\end{lemma}
\begin{proof}
    Fix $n\in\mathbb N$ and a dyadic rational $x=\ell/L$ with $L=2^{k}$. For any 
$n \in \mathbb{N}$, we need to produce a Gaussian integer $v_{n,L,\ell}$ such that
$|f(x) - v_{n,L,\ell}/2^n| \leq  2^{-n}$
in time polynomial in $n + \log L$.

By Definition~\ref{d:Two}c), there are Gaussian integers $h_{S_0,s,m}$
with \begin{equation}\label{sum_est}\frac{1}{S_0(m)} \sum_{s=0}^{S_0(m)-1} 
|\bar{f}_{S_0(m),s} - h_{S_0(m),s,m}/2^m|^2 \le4^{-m}\end{equation} for $S_0(m)=2^{p(m)}$,
where $\bar f_{S,s}=S\int_{s/S}^{(s+1)/S}f(t)\,dt$. 

Let $C_1=\operatorname{Lip}(f)<\infty$. For every $S\geq  1$, $s<S$, and $x\in[s/S,(s+1)/S)$, we have
\begin{equation} \label{f_hat_est}   
|f(x)-\bar f_{S,s}| = \left| S \int_{s/S}^{(s+1)/S} (f(x) - f(y))\,dy \right|
\le S \int_{s/S}^{(s+1)/S} C_1 |x-y|\,dy
\le \frac{C_1}{S}.
\end{equation}
This is uniform in $S$ and $x$.

Since each summand in \eqref{sum_est} is at most the total sum, 
\begin{equation}\label{hat_h_est}
\bigl|\bar f_{S_0,s}-h_{S_0,s,m}/2^m\bigr|\leq  2^{-m}\sqrt{S_0(m)}=2^{-m+p(m)/2}
\qquad\text{for all }s<S_0.
\end{equation}

To obtain $2^{-m+p(m)/2}\leq  2^{-n-3}$  at the right-hand side of \eqref{hat_h_est}), we need to have  $2m-p(m)=(2-a)m-b\geq 2n+6$ or, equivalently, $m\geq\frac{2n+6+b}{2-a}$. 

To obtain $\frac{C_1}{S_0(m)}\leq 2^{-n-3}$ at the right-hand side of \eqref{f_hat_est},  we need $p(m)\geq n+3+\log_2C_1$ or, equivalently, $m \geq (n+3+\log_2 C_1 - b)/a$.

Set
\begin{equation}\label{m(n)}    
m:=\max\left\{\Bigl\lceil \frac{2n+6+b}{2-a}\Bigr\rceil,\quad \left\lceil \frac{n+3+\log_2 C_1 - b}{a}\right\rceil\right\},
\end{equation}
then we get
\begin{equation}\label{f_h_est}
|f(x)-h_{S_0(m),s,m}/2^m|\leq  2^{-n-2}.
\end{equation}
Set $v_{n,L,\ell}:=\lfloor h_{S_0(m),s,m}\cdot 2^{n-m}\rceil\in\mathbb Z[i]$, with $m$ chosen as in \eqref{m(n)}. Then $|v_{n,L,\ell}/2^n-h_{S_0(m),s,m}/2^m|\leq  \frac{\sqrt{2}}{2}2^{-n}$, so, from \eqref{f_h_est}, we obtain the desired estimate
\begin{equation}\label{f_v_est}
|f(x)-v_{n,L,\ell}/2^n|\leq  2^{-n}.
\end{equation}

The index $s=\lfloor S_0(m)x\rfloor$ is computable in time $\operatorname{poly}(p(m)+k)=\operatorname{poly}(n+\log L)$.
The value $h_{S_0(m),s,m}$ is computable in time $\operatorname{poly}(m+\log S_0(m))=\operatorname{poly}(m+p(m))=\operatorname{poly}(n)$ by Definition \ref{d:Two} c(ii), since $m=O(n)$ and $p(m)=am+b\leq  2m$ for $m\geq\frac{b}{2-a}$.
The rounding takes time $\operatorname{poly}(n+\log L)$.
Hence $v_{n,L,\ell}$ is computable in time $\operatorname{poly}(n+\log L)$. $\qed$

\end{proof}

\begin{proof} [Proposition \ref{t:heat_step}]
    By Theorem \ref{t:StepMeanNotOther}, there exists a function $f\in \ELLtwo[0,1]$ that is polynomial time Step-computable in mean, and for which  $\bar{f}_{1,0}=\hat{f}_0\in\mathbb C$ is not polynomial-time computable unless $\sharp P_1\subseteq FP_1$.
    
  %  \Sveta{[TBD: add the explanation why we can take $S_0(m)=2^m$ by adjusting the witness from the proof of Theorem \ref{t:StepMeanNotOther}: pick p(m)=m instead of $\sqrt{m}/2$]} 

%At scale $2^m$, the witness of Theorem \ref{t:StepMeanNotOther}
%is resolved exactly on levels $m'\le m/2$, while levels $m'>m/2$
%contribute squared mean error at most $2^{-m^2/2-O(m)}\le 4^{-m}$
%for $m$ large, so the witness is step-computable in polynomial time on average with
%$S_0(m)=2^m$. The cell values at scale $2^m$ are computable in
%polynomial time, since each cell either lies inside a level-$m'$
%interval with $m'\le m/2$ (height read off explicitly) or in a finer
%level $m'>m/2$ (height assigned $0$).
 
 Note that, by well-known properties of the heat equation, for any $t$, the solution $u(\cdot,t)$ is analytic, and
 \begin{equation}\label{heat1}
     \int_0^1 u(x,t)\,dx = \int_0^1 f(x)\,dx = \hat{f}_0. 
 \end{equation}
Suppose that $u(\cdot,t)$ is Step-computable in polynomial time on average with $S_0(m)=2^{p(m)}$, where $p(m)=am+b$ with $0<a<2$.  Since it is analytic, by Lemma \ref{l:analyt-stepmean}, it is also polynomial-time computable,
hence $\int_0^1 u(x,t)\,dx$ is polynomial-time computable.

 By \eqref{heat1}, $\int_0^1 u(x,t)\,dx=\hat{f}_0$.  
Therefore $\hat{f_0}$ is polynomial-time computable, contradiction. \qed   
\end{proof}

\begin{remark}\label{r:heat-step-mean}
 1.   The condition  $p(m)=am+b$ with $0<a<2$ in Lemma~\ref{l:analyt-stepmean} (and hence in Proposition \ref{t:heat_step} for $u$) is not (at least straightforwardly) relaxed by assuming $f$
analytic. The estimate $O(1/S)$ in \eqref{f_hat_est} of its proof is tight for analytic functions
(e.g.\ $f(x)=e^{2\pi ix}$ at $x=0$). The estimate  $2^{-m+p(m)/2}$ in \eqref{hat_h_est} is
tight for any $f$ because the approximation $h$ is only constrained by
the mean bound independent of smoothness of $f$.

2. The function $f$ in Propositions \ref{t:heat_step} can be chosen step-computable in  polynomial time on average with $S_0(m)=2^m$, by the same computation as in the proof of Theorem~\ref{t:StepMeanNotOther}.

3. Note that the analogue of Lemma \ref{l:analyt-stepmean} for Step-computable in polynomial time
(Definition~\ref{d:Two}b) trivially holds by Theorem \ref{t:Continuous2}. 
\end{remark}

\begin{lemma}\label{l:analyt-s-f-equiv}
   Let $g \in \ELLtwo[0,1]$ be analytic and $1$-periodic. Then $g$ being Step-computable in polynomial time(Definition  \ref{d:Two} b), is equivalent to $g$ being Fourier-computable in polynomial time (Definition \ref{d:Two} a).
\end{lemma}

\begin{theorem}
\label{t:equiv-heat}
The following are equivalent:
\begin{itemize}
\item[(i)] Question \ref{q}(1): every  Step-computable in polynomial time
$f \in \ELLtwo[0,1]$ is Fourier-computable in polynomial time.
\item[(ii)] For every polynomial-time computable $t > 0$ and every
Step-computable in polynomial time $f \in \ELLtwo[0,1]$, the solution $u(\cdot,t)$ of \eqref{eq:heat} is
Step-computable in polynomial time.
\end{itemize}
\end{theorem}

\begin{proposition}
\label{p:fourier-to-step-heat}
Let $f \in \ELLtwo[0,1]$ be Fourier-computable in polynomial time. and $t > 0$
polynomial-time computable. Then $u(\cdot,t)$ is Step-computable in polynomial time.
In particular, $u(\cdot,t)$ is Step-computable in polynomial time on average.
\end{proposition}
\begin{proof}
    By Theorem \ref{t:heat}, $u(\cdot,t)$  is Fourier computable in polynomial time. Since it is also analytic, by Lemma \ref{l:analyt-s-f-equiv}, $u(\cdot,t)$ is Step-computable in polynomial time (hence also Step-computable in polynomial time on average).
\end{proof}

\Sveta{[TBD add proofs of Lemma \ref{l:analyt-s-f-equiv} and Theorem \ref{t:equiv-heat}, and adjust statements if necessary - in the next version.]}

%%%%%%%%%%%%%%%%%%%%%%%%%%%%
\subsection{Wave Equation}
\label{ss:Wave}

The $d$-dimensional Wave Equation 
\begin{equation}
\label{e:Wave}
\begin{cases}
u_{t t}(\vec x,t) \;=\; \Delta u(\vec x,t) := \sum\nolimits_{j=1}^d u_{jj}, 
\qquad t > 0, \\   
u(\vec x,0) \;=\; f(\vec x), \quad u_t(\vec x,0) \;=\; g(\vec x) ,
\quad \vec x\in\Omega
\end{cases}
\end{equation}
is notorious for purportedly violating the Church-Turing Hypothesis,
namely as a physical system evolving from computable initial state 
$(f,g)$ at $t=0$
into an uncomputable state $u(\cdot,t)$ 
at a computable time $t>0$ \cite{PR81,PZ97,WZ02a,BP22}.
In dimension $d=1$,
the explicit solution to \eqref{e:Wave} is given by d'Alembert:
\begin{equation}
\label{e:WaveSol}
u(x,t) \quad=\quad \Big( f(x+t) \: +\: f(x-t) \; + \; \int_{x-t}^{x+t} g(y)\,dy \Big)/2  \enspace .
\end{equation}
As in \Cref{ss:Heat}, we consider the compact
unit interval with periodic boundary conditions $\Omega:=[0;1)\bmod 1$.

\begin{theorem}
\label{t:Wave}
\begin{enumerate}
\item[(1)]
If $f,g\in\calC(\Omega)$ are computable in \classSharpP,
then so is $u(\cdot,t)$ for every polynomial-time computable $t>0$.
\item[(2)]
There exist polynomial-time computable $f,g\in\calC^\infty(\Omega)$
such that $u(0,\tfrac{1}{2})$ is not polynomial-time computable 
unless $\classSharpP_1\subseteq\classFP_1$.
\item[(3)]
If $f,g\in\ELLtwo(\Omega)$ are Fourier-computable in polynomial time (a),
then so is $u(\cdot,t)$ for every polynomial-time computable $t>0$.
\end{enumerate}
\end{theorem}

\begin{proof}[\Cref{t:Wave}]
\begin{enumerate}
\item[(1)]
follows from \cref{e:WaveSol} and Fact~\ref{f:ContinuousClosure}.
\item[(2)]
Take $f:\equiv0$ and $g$ from \cite[Theorem~5.32de]{Ko91}.
\item[(3)]
follows from \cref{e:WaveSol} and \Cref{t:FourierClosure}.
\qed\end{enumerate}\end{proof}

%%%%%%%%%%%%%%%%%%%%%%%%%%%%%%%%%%%%%%%%%%%%%%%%
\section{Conclusion}\label{s:Conclusion}

Choosing the `right' basis is well-known computationally beneficial \cite{ClassOfBases}.
In $d$-dimensional Linear Algebra, say, matrix-vector multiplication
in a basis of eigenvectors reduces the arithmetic cost from quadratic to linear.
The present work rigorously confirms a much larger efficiency gain.
%\Sveta{[Should be edited after all the PDE results are finalized.]}
We introduced two possible definitions of %(Definition \ref{d:Two}) of
polynomial time computability of $\ELLtwo$ functions, which we have shown to be not equivalent unless $\classSharpP_1\subseteq\classFP_1$. One of the possible insights to the question which of them is ``more suitable'' could be provided by PDEs, which are a large source of applications of generalized (including integrable) functions.

For both Heat and Wave equations Fourier-polynomial computability turned out to be preserved by the solution operator (these results can be similarly extended to higher dimensions and some more general cases of boundary conditions), which is not  the case for the ``classical'' definition of polynomial-time computability for continuous functions, unless $\classSharpP_1\subseteq\classFP_1$ \cite{KPSZ23}. Step-polynomial computability (and, in particular, Step-polynomial computability on average), which is well compatible with grid or finite element methods, is potentially applicable to an even broader class of PDEs.
%however, it is not preserved by Heat solution operator, unless $\classSharpP_1\subseteq\classFP_1$. %\Sveta{[The latter still needs to be written down and checked.]} 

Further systematic studies of computational complexity of various (classes of systems of) PDEs for the case of generalized solutions, as well as developing corresponding Exact Real Computation methods \cite{DBLP:journals/lmcs/ParkEtAl24}, are of interest for future work. In particular, according to \cite{KPSZ23}, computing "classical" solutions to well posed evolutionary systems of PDEs with variable coefficients is \classPSPACE-complete, unless the coefficients and initial data are analytic functions; how about generalized solutions w.r.t. the definitions of the present paper? Theorem \ref{t:heat} and Theorem \ref{t:Wave}(3) give a hope for efficiency of Exact Real Computation methods for some examples of computing generalized solutions of PDEs. Also, further investigations could extend the presented $\ELLtwo$ polynomial-time computability notions to the case of operators, as in \cite{KC10} for classical computability of continuous functions. 

\begin{remark}
\label{r:Nonuniform}
Some of our proofs/algorithms are non-uniform,
namely require `knowing' an integer upper bound $F$ to $\|f\|_2$:
such as in \Cref{t:SharpP}$\calF\Rightarrow\calS$ 
or in \Cref{t:StepClosure}(2)
\end{remark}

\begin{question}\label{q}
\begin{enumerate}
\item[(1)]
If $f\in\ELLtwo[0;1]$ is Step-computable in polynomial time (b),
is $f$ then also Fourier-computable in polynomial time (a)?
\item[(2)]
Suppose $f\in\ELLtwo[0;1]$ is Step-computable in polynomial time
and $g\in\calC[0;1]$ is polynomial-time computable.
Is $f\cdot g\in\ELLtwo[0;1]$ then necessarily Step-computable in polynomial time again?
\item[(3)]
Is there a definition of polynomial-time computability for $\ELLtwo$ functions
with closure properties (even) better than our Fourier (a)?
\end{enumerate}
\end{question}
Note that (2) is related to (1)
via \Cref{t:StepClosure}(2)
since $g_k(t):=\exp(-2\pi i k t)\in\calC[0;1]$
is polynomial-time computable.

\subsubsection*{Acknowledgements.}
A. Bacho acknowledges support from the Air Force Office of Scientific Research (AFOSR) under awards FA9550-20-1-0358 (MURI: “Machine Learning and Physics-Based Modeling and Simulation”) and FA9550-24-1-0237, the U.S. Department of Energy (DOE), Office of Science, Office of Advanced Scientific Computing Research (ASCR) under award DE-SC0023163 (SEA-CROGS: “Scalable, Efficient and Accelerated Causal Reasoning Operators, Graphs and Spikes for Earth and Embedded Systems”), and the Office of Naval Research (ONR) under award N00014-25-1-2035. 
S. Selivanova is supported by PJSC ``Gazprom Neft'', aggr. No GPN-26/03000/00502/R.
%S. Selivanova is supported by the Mathematical Center in Akademgorodok under the agreement No 075-15-2025-349 with the Ministry of Science and Higher Education of the Russian Federation.
M. Ziegler is supported by KAIST in-house grant KC30.
We thank Donghyun Lim for promoting this research direction
\url{http://youtu.be/KAQi8kdwJM0}

%%%%%%%%%%%%%%%%%%%%%%%%%%%%%%%%%%%%%%%%%%%%%%%%
\addcontentsline{toc}{section}{References}
\bibliographystyle{plain}
\bibliography{cca,hilbert}

\end{document}